\documentclass{article}

\usepackage[preprint]{neurips_2024}

\usepackage[utf8]{inputenc} 
\usepackage[T1]{fontenc}    
\usepackage{hyperref}       
\usepackage{url}            
\usepackage{booktabs}       
\usepackage{amsfonts}       
\usepackage{nicefrac}       
\usepackage{microtype}      
\usepackage{xcolor}         

\usepackage{amsmath,amsfonts,bm}

\def\eqref#1{equation~\ref{#1}}

\def\1{\bm{1}}

\def\ru{{\textnormal{u}}}

\def\rvg{{\mathbf{g}}}

\def\rvk{{\mathbf{k}}}

\def\rvx{{\mathbf{x}}}
\def\rvy{{\mathbf{y}}}

\def\vp{{\bm{p}}}

\DeclareMathAlphabet{\mathsfit}{\encodingdefault}{\sfdefault}{m}{sl}
\SetMathAlphabet{\mathsfit}{bold}{\encodingdefault}{\sfdefault}{bx}{n}

\newcommand{\E}{\mathbb{E}}

\newcommand{\R}{\mathbb{R}}

\DeclareMathOperator*{\argmax}{arg\,max}

\usepackage{hyperref}
\usepackage{url}
\usepackage{amsmath}
\usepackage{enumitem}
\setlist[itemize]{itemsep=2pt, parsep=2pt, topsep=0pt, partopsep=0pt}
\usepackage{CJKutf8}
\usepackage{graphicx}
\usepackage{booktabs}
\usepackage{array}
\usepackage{graphicx}
\usepackage{listings}
\usepackage{tabularx}  
\usepackage{tablefootnote}  
\usepackage{caption}  
\usepackage{subcaption}
\usepackage{xcolor}
\definecolor{codegreen}{rgb}{0,0.6,0}
\definecolor{codegray}{rgb}{0.5,0.5,0.5}
\definecolor{codepurple}{rgb}{0.58,0,0.82}
\definecolor{backcolour}{rgb}{0.95,0.95,0.92}
\lstdefinestyle{mystyle}{
    backgroundcolor=\color{backcolour},   
    commentstyle=\color{codegreen},
    keywordstyle=\color{magenta},
    numberstyle=\tiny\color{codegray},
    stringstyle=\color{codepurple},
    basicstyle=\ttfamily\footnotesize,
    breakatwhitespace=false,         
    breaklines=true,                 
    captionpos=b,                    
    keepspaces=true,                 
    numbers=left,                    
    numbersep=5pt,                  
    showspaces=false,                
    showstringspaces=false,
    showtabs=false,                  
    tabsize=2,
    inputencoding=utf8,
    escapeinside={(*@}{@*)}, 
}
\usepackage{multirow}  
\usepackage{tcolorbox}
\usepackage{wrapfig}
\usepackage{algorithm}
\usepackage{algpseudocode}
\usepackage{amsthm}

\newtheorem{proposition}{Proposition}[section]
\newtheorem{lemma}{Lemma}[section]

\title{WaterPool: A Watermark Mitigating Trade-offs among Imperceptibility, Efficacy and Robustness}

\author{%
  Baizhou Huang,~~~ Xiaojun Wan\\
  Wangxuan Institute of Computer Technology, Peking University\\
  The MOE Key Laboratory of Computational Linguistics, Peking University\\
  \texttt{hbz19@pku.edu.cn, wanxiaojun@pku.edu.cn} \\
}

\begin{document}

\maketitle

\vspace{-15pt}
\begin{abstract}
With the increasing use of large language models (LLMs) in daily life, concerns have emerged regarding their potential misuse and societal impact.
Watermarking is proposed to trace the usage of specific models by injecting patterns into their generated texts. 
An ideal watermark should produce outputs that are nearly indistinguishable from those of the original LLM (imperceptibility), while ensuring a high detection rate (efficacy), even when the text is partially altered (robustness).
Despite many methods having been proposed, none have simultaneously achieved all three properties, revealing an inherent trade-off.
This paper utilizes a key-centered scheme to unify existing watermarking techniques by decomposing a watermark into two distinct modules: a key module and a mark module. Through this decomposition, we demonstrate for the first time that the key module significantly contributes to the trade-off issues observed in prior methods. Specifically, this reflects the conflict between the scale of the key sampling space during generation and the complexity of key restoration during detection.
To this end, we introduce \textbf{WaterPool}, a simple yet effective key module that preserves a complete key sampling space required by imperceptibility while utilizing semantics-based search to improve the key restoration process. WaterPool can integrate with most watermarks, acting as a plug-in. Our experiments with three well-known watermarking techniques show that WaterPool significantly enhances their performance, achieving near-optimal imperceptibility and markedly improving efficacy and robustness (+12.73\% for KGW, +20.27\% for EXP, +7.27\% for ITS).
\end{abstract}

\vspace{-10pt}
\section{Introduction}

\begin{wrapfigure}[16]{r}{0.4\textwidth}
\vspace{-20pt}
  \centering
  \includegraphics[width=0.95\linewidth]{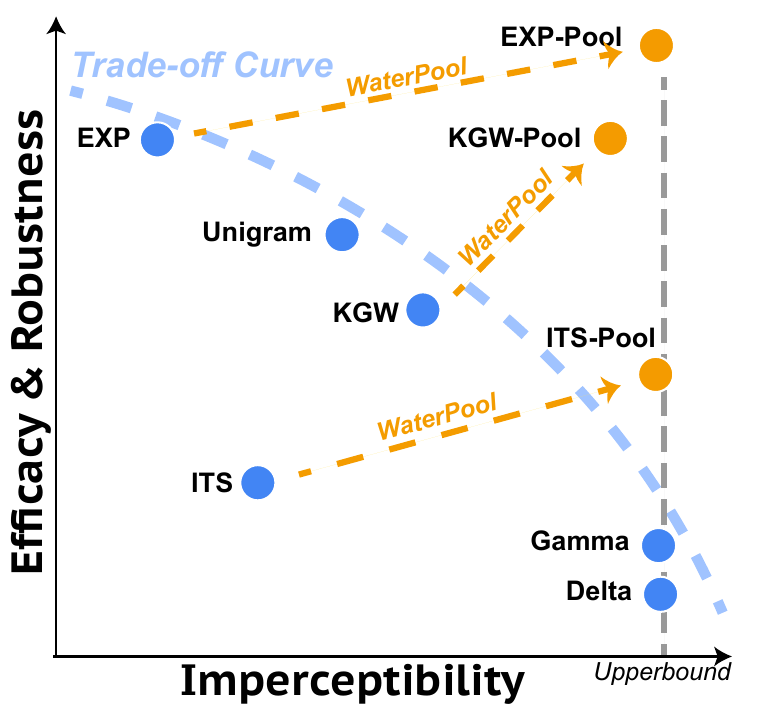}
  \caption{\small Previous methods made trade-offs among imperceptibility, efficacy and robustness. WaterPool mitigates this problem and improve KGW, ITS, EXP significantly.}
  \label{fig:intro-trade-off-performance}
\end{wrapfigure}

The world has recently witnessed the great power of large language models (LLMs). Models like OpenAI's ChatGPT have been widely integrated into daily life. However, the widespread use of these systems has raised significant concerns about their potential misuse. For example, LLMs could be used to generate massive amounts of fake news or automated comments to manipulate social media, posing threats to academic integrity and intellectual property rights \cite{10.1145/3442188.3445922,liuSurveyTextWatermarking2023}.


To address these issues, watermarking has been proposed to track the usage of specific models \cite{kirchenbauerWatermarkLargeLanguage2023}. An ideal watermark embeds an invisible pattern within generated contents of an LLM by sampling outputs from a stochastic modified distribution. The expectation of the modified distribution is nearly identical to the original one, making the watermarked text almost indistinguishable from the original (\textbf{imperceptibility}). But at the same time, this pattern can be reliably detected by statistical detector with a low false-positive rate (\textbf{efficacy}) and remains detectable even if the text is corrupted by semantic-preserving attacks (\textbf{robustness}).


Despite significant progress being made in prior works, achieving all three properties simultaneously has proven challenging, as illustrated in Figure \ref{fig:intro-trade-off-performance}. It is widely accepted that there is a trade-off among imperceptibility, efficacy, and robustness \citep{ajithPerformanceTradeoffsWatermarking2023,molendaWaterJudgeQualityDetectionTradeoff2024}. Previous methods often use hyper-parameters to balance this trade-off, like the $\delta$ in KGW \citep{kirchenbauerWatermarkLargeLanguage2023} controlling the degree of distribution shift.


In this paper, we provide a key-centered scheme to review and unify typical watermarking techniques. The scheme decomposes a watermarking technique into two independent modules, a key module and a mark module, as shown in Figure \ref{fig:intro-framework}. During generation, the key module samples a private key to provide a source of randomness. It is then utilized by the mark module as a seed to modify the next token distribution, from which watermarked texts are sampled. During detection, the key module attempts to restore the possible key from the given candidate text. Then the mark module aligns the text with the restored key to compute statistics, which imply the likelihood of watermark presence. Under the decomposition, we separate the requirements of imperceptibility, efficacy and robustness into the two modules. Building on this, we for the first time demonstrate that the key module significantly contributes to the trade-off problem. Specifically, it reflects the conflict between the scale of key sampling space during generation and the complexity of key restoration during detection.

To overcome this trade-off, we introduce \textbf{WaterPool}, a simple but effective key module. WaterPool maintains the complete key sampling space, crucial for imperceptibility, while leveraging a semantics-based search to significantly enhance the precision and effectiveness of the key restoration process, thereby ensuring high robustness against attacks. We integrate WaterPool into three of the most renowned watermarking techniques, EXP \citep{kuditipudiRobustDistortionfreeWatermarks2023}), KGW \citep{kirchenbauerWatermarkLargeLanguage2023} and ITS \citep{kuditipudiRobustDistortionfreeWatermarks2023}). WaterPool effectively mitigates the traditionally "inevitable" trade-offs, achieving superior performance as shown in Figure \ref{fig:intro-trade-off-performance}. 

Our experiments include two scale of large language models (LLMs) across tasks of open-ended generation and long-form question answering. Experimental results demonstrate the supreme capabilities of our proposed WaterPool. On one hand, it elevates the imperceptibility of KGW, EXP and ITS to near-optimal levels. On the other hand, it significantly enhances the efficacy and robustness of previous watermarking techniques, yielding substantial improvements across different experimental settings (+12.73\% for KGW, +20.27\% for EXP, +7.27\% for ITS).

\begin{figure*}[!tb]
    \centering
    \includegraphics[width=\linewidth]{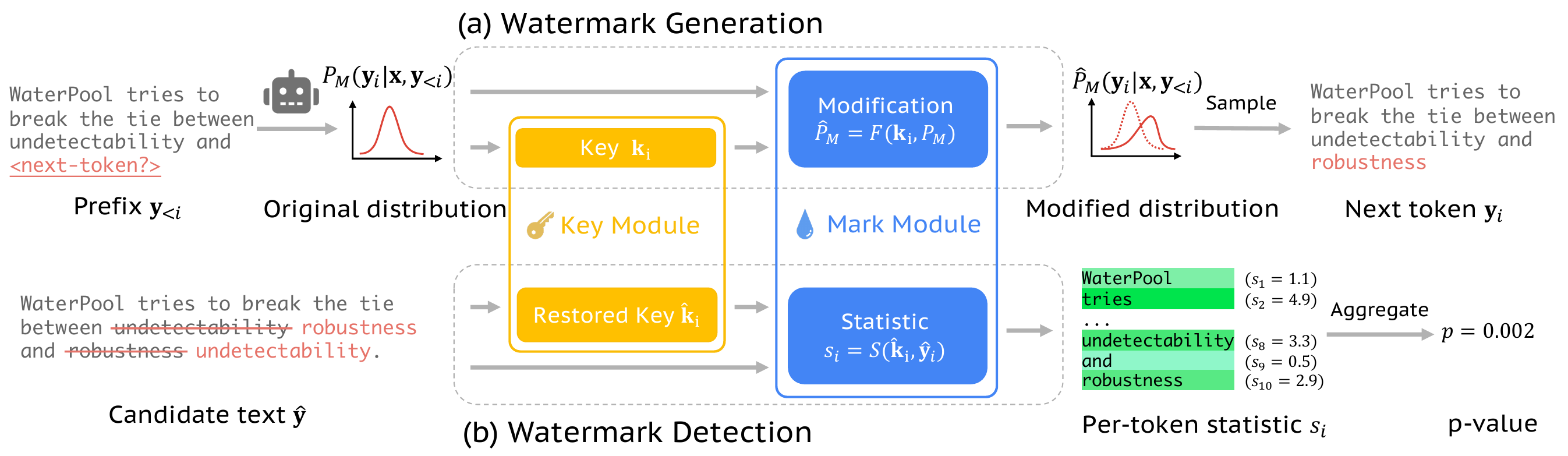}
    \vspace{-10pt}
    \caption{\small Overview of key-centered watermarking scheme. A watermark is decomposed into two modules, a key module and a mark module. (a) During generation, the LLM provides an next token distribution $P_M$. The key module samples a private key $\rvk_i$ as a random seed for the mark module to stochastically modify the distribution to $\hat P_M$, from which watermarked texts are sampled. (b) During detection, the key module restores the key $\hat \rvk_i$ for each candidate token. The mark module then calculates the per-token statistic $s_i$ based on the restored key and aggregates them for $p$-value.}
    \label{fig:intro-framework}
\end{figure*}

\section{Preliminary}
\label{sec:problem-formulation}
\label{sec:three-properties}


\paragraph{Problem formulation.}
We begin by formalizing the process of watermarking. Given any prompt $\rvx$, a LLM will generate a sequence of outputs $\rvy_i\sim P_M(\cdot|\rvx,\rvy_{<i})$ in an auto-regressive manner, where $P_M(\cdot |\rvx,\rvy_{<i})\in\Delta(\Sigma)$ is the $i$-th token distribution over the vocabulary $\Sigma$. A watermark invokes in generation via stochastically modifying the distribution to a new distribution $\hat P_M(\cdot |\rvx,\rvy_{<i})$.
The detection of watermarked texts is formulated as a hypothesis testing problem, where the alternative hypothesis is that the candidate is sampled from a modified distribution. It is typically proved by gathering per-token statistics $s_i$ for a one-tailed significance test.


\paragraph{Imperceptibility.} An ideal watermark should maintain the output distribution of LLM as unchanged as possible \citep{christUndetectableWatermarksLanguage2023,huUnbiasedWatermarkLarge2023,kuditipudiRobustDistortionfreeWatermarks2023,wuDiPmarkStealthyEfficient2023}. It is both common and practical of watermarks to work as a plug-in to LLMs. Therefore, the shift in the output distribution should be minimal to preserve the inherent properties of the original model. Moreover, significant different patterns, such as the frequent occurrence of specific tokens in Unigram \citep{zhaoProvableRobustWatermarking2023} may induce easier attacks. Formally, we define the imperceptibility following "$N$-shot undetectable" from \citet{huUnbiasedWatermarkLarge2023},
\vspace{-7pt}
\begin{equation}
    \prod_n^N P_M(\rvy^n |\rvx^n) = \E_{\hat P_M}[ \prod_n^N \hat P_M(\rvy^n|\rvx^n)] \quad \forall \rvx^n,\rvy^n \in \Sigma^*
\end{equation}

\vspace{-10pt}
Here, superscripts $n$ indicates different rounds of generation. For LLMs, which model language in an auto-regressive manner, the equation above can be expressed as:
\vspace{-7pt}
\begin{equation}
    \label{eq:def-imperceptibility}
    \prod_{i,n} P_M(\rvy_i |\rvx^n,\rvy_{<i}^n) = \E_{\hat P_M}[\prod_{i,n} \hat P_M(\rvy_i |\rvx^n,\rvy_{<i}^n)] \quad \forall \rvx^n,\rvy^n \in \Sigma^*
\end{equation}

\vspace{-10pt}
We want to highlight the importance of the product over multiple generations above. It indicates that it is infeasible to distinguish between the original and the watermarked texts without priori knowledge about the modified distribution, even when many rounds of queries are allowed \citep{christUndetectableWatermarksLanguage2023}.

\vspace{-10pt}
\paragraph{Efficacy.} An ideal watermark technique should be able to distinguish watermarked texts from the others. Empirically, it is required to achieve high true positive rate with low false positive rate. Most watermarking techniques achieve this by ensuring a substantial difference between the per-token statistic under alternative hypothesis ($H_1$: the candidate token is sampled from the modified distribution $\hat P_M(\cdot |\rvx^n,\rvy^{n}_{<i})$) and null hypothesis ($H_0$: the candidate token is sampled from other distributions). This can be formulated as,
\vspace{-4pt}
\begin{equation}
    \label{eq:def-efficacy}
    \E[s_i|H_1]-\E[s_i|H_0] \geq \phi(\vp^i)
\end{equation}

\vspace{-7pt}
where $s_i$ is the statistic of $i$-th token. We denote the left-hand side as \textit{statistical difference}. $\phi(\vp^i)$ only depends on $\vp^i$, the probability vector for $i$-th step distribution $P_M(\cdot |\rvx,\rvy_{<i})$, remaining constant given $\rvx$ and $\rvy_{<i}$. It indicates the innate potential for watermark injection\footnote{$\phi(\vp)$ has different forms and names in previous works. \citet{kuditipudiRobustDistortionfreeWatermarks2023} defines it as watermark potentials. \citet{kirchenbauerWatermarkLargeLanguage2023} connects it with a special form of entropy.}. 
The statistical difference enables the application of statistical tests like permutation tests \citep{kuditipudiRobustDistortionfreeWatermarks2023} or parametric tests assuming specific distribution forms under the null hypothesis \citep{kirchenbauerWatermarkLargeLanguage2023,fernandezThreeBricksConsolidate2023}.

\vspace{-10pt}
\paragraph{Robustness.} Robustness is a one-step-further requirement of efficacy. 
The statistical pattern of watermarked texts could be vulnerable to potential attacks, including lexical modification or paraphrasing \citep{krishnaParaphrasingEvadesDetectors2023,kirchenbauerReliabilityWatermarksLarge2023}. An ideal watermark technique should be resilient to removal and should maintain high efficacy even after such semantics-preserved attacks.

\section{WaterPool: a Key Module Mitigating Trade-offs}
\vspace{-5pt}

\subsection{Decomposition of Watermark: a Key-centered Scheme}
\vspace{-5pt}


As stated in Section \ref{sec:problem-formulation}, the critical part of watermarking falls in the modified distribution, which is stochastic and determined by a random seed, i.e. the \textit{private key} \citep{christUndetectableWatermarksLanguage2023,kuditipudiRobustDistortionfreeWatermarks2023}. With private keys as connection, we decompose watermarks into two independent modules: a key module and a mark module. The key module handles the sampling and restoration of private keys, while the mark module is responsible for the modification process and per-token statistic. 

Specifically, during the $i$-th step of generation, the key module samples a private key $\rvk_i$ from the possible key space $\Xi\subset\R$ to provide randomness. Then the mark module takes the sampled key as a random seed to modify the next token distribution $P_i:=P_M(\cdot |\rvx^n,\rvy_{<i}^n)$ to $F(\rvk_i,P_i)$, where $F:\R\times\Delta(\Sigma)\rightarrow\Delta(\Sigma)$ is a stochastic function. For example, KGW \citep{kirchenbauerWatermarkLargeLanguage2023} increases the probabilities of randomly selected tokens by adding a constant to their logits.

During detection, a candidate text $\hat \rvy$ is given. Watermarking techniques generally frame the detection as a hypothesis testing problem, treating each token as an independent sample. For each token $\hat\rvy_i$, the key module tries to restore the corresponding private key $\hat \rvk_i$ used in generation based on the context. The mark module then calculates the per-token statistic $s_i=S(\hat\rvy_i, \hat\rvk_i)$, where $S:\R\times\Sigma\rightarrow\R$. 
These statistics are aggregated to obtain an overall score, indicating the confidence that the entire sequence is watermarked. 

In this scheme, the key module and the mark module operate independently, with private keys serving as the only connection. This independence allows for the combination of any key module with any mark module to create new watermarking methods. We review several well-known watermarking techniques and unify them into this scheme. The common designs of mark modules and key modules are presented in Table \ref{tab:prior-watermark-design-mark} and \ref{tab:prior-watermark-design-key}.

\begin{table}[tb]
  \vspace{-10pt}
  \newlength{\twidth}
  \setlength{\twidth}{6cm}
  \caption{
  Mark modules of typical watermarks. 
  $P_i(\cdot)$, $C_i(\cdot)$ and $L_i(\cdot)$ are probability distribution function, cumulative distribution function and logit function of the $i$-th step generation. 
  $ord: \Sigma\rightarrow \{1,...,|\Sigma|\}$ is a function mapping each token to its order in the vocabulary. 
  $\delta(t)$ represents a degeneration distribution taking only one value $t$. 
  Subscript "$\rvk_i$" indicates a random variable seeded by $\rvk_i$. For example, $\mathcal{G}_{\rvk_i}$ is a random $\gamma$ ratio partition of vocabulary, $\pi_{\rvk_i}: \Sigma\rightarrow \Sigma$ represents a random permutation over the vocabulary. 
  Superscript "$perm$" indicates a distribution function after vocabulary permutation, such as $P^\text{perm}$.
  Other notations like $\gamma, \tau$ are fixed hyper-parameters.
  }
  \vspace{3pt}
  \centering
  \makebox[\textwidth][c]{
  \resizebox{1\linewidth}{!}{
  \begin{tabular}{cccc}
    \toprule
    Name & Reweight Function $F(\rvk_i, P_i)$ & Statistic $S(\hat\rvk_i,\hat\rvy_i)$ & Used in\\
    \midrule
    logits-add & 
    \begin{minipage}[c]{\twidth}
    \scriptsize
        $ \forall t\in \Sigma, \hat L_i(t) = L_i(t) + \tau\cdot \1_{t\in \mathcal G_{\rvk_i}}$ \\
        $F(\rvk_i, P_i) = \text{Softmax}(\hat L_i)$
    \end{minipage}
    & 
    \begin{minipage}[c]{\twidth}
        $S(\hat\rvk_i,\hat\rvy_i)=\frac{\1_{\hat\rvy_i\in \mathcal G_{\hat\rvk_i}} - \gamma}{\sqrt{\text{len}(\rvy)\gamma(1-\gamma)}}$
    \end{minipage}
    &
    KGW, Unigram
    \\
    \midrule
    inverse-sample & 
    \begin{minipage}[c]{\twidth}
    \scriptsize
        $\ru_{\rvk_i}\sim\mathcal U([0,1])$\\
        $\forall t\in\Sigma, P^{\text{perm}}_i(t)=P_i(\pi_{\rvk_i}^{-1}(t))$\\
        $z^*=\min\{z\in\mathbb{Z}: \sum_{\text{ord}(t)\leq z} P_i^{\text{perm}}(t) > \ru_{\rvk_i}\}$\\
        $t^*=\pi_{\rvk_i}^{-1}\circ\text{ord}^{-1}(z^*)$\\
        $F(\rvk_i,P_i)=\delta(t^*)$
    \end{minipage}
    & 
    \begin{minipage}[c]{\twidth}
        $S(\hat\rvk_i,\hat\rvy_i)=(\ru_{\hat\rvk_i}-\frac12)\cdot(\frac{\text{ord}(\pi_{\hat\rvk_i}(\hat\rvy_i))-1}{|\Sigma|-1}-\frac12)$
    \end{minipage}
    &
    ITS
    \\
    \midrule
    gumbel-sample & 
    \begin{minipage}[c]{\twidth}
    \scriptsize
        $\rvg_{\rvk_i}\in\R^{|\Sigma|}, \rvg_{\rvk_i,j}\overset{\mathrm{iid}}{\sim} \text{Gumbel(0,1)}$\\
        $t^*=\arg\max_{t\in\Sigma} \log P_i(t)+\rvg_{\rvk_i,\text{ord}(t)}$\\
        $F(\rvk_i,P_i)=\delta(t^*)$
    \end{minipage}
    & 
    \begin{minipage}[c]{\twidth}
        $S(\hat\rvk_i,\hat\rvy_i)=-\exp(-\rvg_{\rvk_i,\text{ord}(\hat\rvy_i)})$
    \end{minipage}
    &
    EXP
    \\
    \midrule
    prob-scale & 
    \begin{minipage}[c]{\twidth}
    \scriptsize
        $\forall t\in\Sigma, P^{\text{perm}}_i(t)=P_i(\pi_{\rvk_i}^{-1}(t))$\\
        $\forall t\in\Sigma, C^{\text{perm}}_i(t)=\sum_{\text{ord}(t')<\text{ord}(t)} P^{\text{perm}}_i(t')$ \\
        $\forall t\in\Sigma, \hat C_i(t)=\min(2C^{\text{perm}}_i(t), 1)$\\
        $\forall t\in\Sigma, \hat P^{\text{perm}}_i(t)=\hat C_i(t)-\hat C_i(\text{ord}^{-1}(\text{ord}(t) - 1))$\\
        $\forall t\in\Sigma, \hat P_i(t)=\hat P^{\text{perm}}_i(\pi_{\rvk_i}(t))$\\
        $F(\rvk_i,P_i)=\hat P_i$
    \end{minipage}
    & 
    \begin{minipage}[c]{\twidth}
    \scriptsize
        $\forall t\in\Sigma, P^{\text{perm}}_i(t)=P_i(\pi_{\hat\rvk_i}^{-1}(t))$\\
        $\forall t\in\Sigma, C^{\text{perm}}_i(t)=\sum_{\text{ord}(t')<\text{ord}(t)} P^{\text{perm}}_i(t')$ \\
        $\forall t\in\Sigma, \hat C_i(t)=\min(2C^{\text{perm}}_i(t), 1)$\\
        $\forall t\in\Sigma, \hat P^{\text{perm}}_i(t)=\hat C_i(t)-\hat C_i(\text{ord}^{-1}(\text{ord}(t) - 1))$\\
        $\forall t\in\Sigma, \hat P_i(t)=\hat P^{\text{perm}}_i(\pi_{\hat\rvk_i}(t))$\\
        $S(\hat\rvk_i,\hat\rvy_i)=\log \hat P_i(\hat\rvy_i) - \log P_i(\hat\rvy_i)$
    \end{minipage}
    & 
    Gamma
    \\
        
    \bottomrule
  \end{tabular}
  }
  }
  \label{tab:prior-watermark-design-mark}
\vspace{-10pt}
\end{table}

\begin{table}[tb]
  \newlength{\twidthInKeyTable}
  \setlength{\twidthInKeyTable}{6cm}
  \caption{
  Key modules of typical watermarks. 
  Both $\Xi$ and $k$ are initialized once, and then fixed on every generation.
  $L$ is the maximum length of LLM. $N, c, \eta$ are all fixed hyper-parameters.
  }
  \vspace{3pt}
  \centering
  \makebox[\textwidth][c]{
  \resizebox{\linewidth}{!}{
  \begin{tabular}{cccc}
    \toprule
    Name & Key Sampling $\rvk_i$ & Key Restoration $\hat\rvk_i$ & Used in \\
    \midrule
    greedy-search & 
    \begin{minipage}[c]{4.3cm}
        $\Xi = \{\xi^1,...,\xi^K\}\overset{i.i.d}{\sim} \mathcal U(\R^L)$\\
        $\rvk\sim \mathcal U(\Xi)$
    \end{minipage}
    & 
    \begin{minipage}[c]{\twidthInKeyTable}
        {
        \small
        \begin{align*}
            \hat\rvk = \arg\max_{\xi\in\Xi} d_\text{edit}(&\xi,\hat\rvy), \text{ where}\\
            d_\text{edit}(\xi,\hat\rvy)=\max\{&d_\text{edit}(\xi_{2:},\hat\rvy_{2:}) + S(\xi_1,\hat\rvy_1),\\
            &d_\text{edit}(\xi,\hat\rvy_{2:}) -\eta,\\
            &d_\text{edit}(\xi_{2:},\hat\rvy) -\eta\}
        \end{align*}
        }
    \end{minipage}
    & ITS, EXP
    \\
    \midrule
    context-hash & 
    \begin{minipage}[c]{4.3cm}
        $\rvk_i=\text{hash}(\rvy_{i-c:i-1})$
    \end{minipage}
    & 
    \begin{minipage}[c]{\twidthInKeyTable}
        $\,\rvk_i=\text{hash}(\hat\rvy_{i-c:i-1})$
    \end{minipage}
    &
    KGW, Gamma, Delta
    \\
    \midrule
    fixed-constant & 
    \begin{minipage}[c]{4.3cm}
    \small
        $k\sim \mathcal U(\R);\,\rvk_i=k$
    \end{minipage}
    & 
    \begin{minipage}[c]{\twidthInKeyTable}
        $\,\rvk_i=k$
    \end{minipage}
    &
    Unigram
    \\
    \bottomrule
  \end{tabular}
  }
  }
  \label{tab:prior-watermark-design-key}
\vspace{-5pt}
\end{table}

\subsection{Behind Trade-offs: Conflicts within Key Module}
\label{sec:trade-off-under-key-centered-framework}

In this section, we examine how the key module affects imperceptibility, efficacy and robustness. 
As defined in Section \ref{sec:three-properties}, the statistical difference between the null and alternative hypotheses is crucial for ensuring efficacy and robustness, i.e., whether the candidate token is sampled from a modified distribution. 
Regardless of the choice of statistic $S(\hat\rvk_i,\hat\rvy)$, prior knowledge of the modified distribution is essential.
Since the modification process is stochastic and determined by the private key, successful detection requires that the restored key $\hat\rvk_i$ matches the key used during generation.


The key restoration process fundamentally involves searching through the potential key space $\Xi$.
Both the ITS and EXP methods use a greedy search strategy, enumerating each potential key to identify the one exhibiting the highest statistic. While reliable in key restoration, it is markedly time-consuming. Moreover, the per-token statistic is now $\tilde s=\max s$, potentially diminishing the statistical difference. To mitigate the issue, ITS and EXP limit the possible key space size. To the extreme, Unigram directly fixes the private key.
Alternatively, methods like KGW, Delta, and Gamma take the context window through a hash function for key restoration, reducing time complexity compared to exhaustive searches. But there are risks of incorrect key restoration if the context is altered by attacks, diminishing their robustness.

For imperceptibility under the key-centered scheme, Equation \ref{eq:def-imperceptibility} can be rewritten as:
\vspace{-2pt}
\begin{equation*}
    \prod_{i,n} P_M(\rvy^n_i |\rvx^n,\rvy_{<i}^n) = \E_{\rvk^1,...,\rvk^N}[\prod_{i,n} F(\rvk_i^n,P_i^n)(\rvy_i^n)] \quad \forall \rvx^n,\,\rvy^n\in\Sigma^*
\end{equation*}

\vspace{-5pt}
We propose the following requirements for imperceptibility:
\begin{proposition}
\label{prop:imperceptiblity-requirements}
    A watermark is imperceptible if\, (1) Independent condition: the sampled private key vectors for each generated output are mutually independent, i.e. $\rvk^1,...,\rvk^N\overset{i.i.d}{\sim}\mathcal U(\R^L)$\footnote{$L$ is the maximum output length of LLM.}; (2) Unbiased condition: the modification function $F$ satisfies $P_M(\cdot |\rvx^n,\rvy_{<i}^n) = \E_{\rvk_i\sim\mathcal U(\R)}[F(\rvk_i,P_i)]$.
\end{proposition}

The proposition describes separate requirements for key and mark modules. 
While many mark modules meet the unbiased condition, ensuring independent condition is challenging for key modules.
The independence of private keys over the whole space $\R$ in successive generations indicates the search space $\Xi$ of key restoration grows linearly with the number of generations, as all previously used keys must be considered during detection.

The conflict within key modules now becomes apparent. Imperceptibility requires a large possible key space $\Xi$, which complicates the key restoration process during detection, thereby hindering both efficacy and robustness. This trade-off within the key module is a critical factor underlying the broader trade-off among the three properties.

\subsection{WaterPool: a Semantics-based Key Module}

The previous section highlights a conflict within the key module, which contributes to the trade-off among imperceptibility, efficacy and robustness in watermarking. A natural question arises: is it possible to break the conflict for mitigating the trade-off? To ensure imperceptibility, it is essential to maintain a sufficiently large key space, characterized by an i.i.d. sampling strategy. Thus, an efficient and precise search strategy for the private key is needed for efficacy. The context-hash strategy provides insights by leveraging contextual information, yet hash functions are fragile to attacks. We aim to find a robust signal within the candidate's context that withstands semantic-preserved attacks.

To this end, we propose \textbf{WaterPool}, a simple but effective key module empowered by semantic searching. Specifically, for each generation, WaterPool independently samples a private key vector $\rvk^n\sim\mathcal U(\R^L)$ to meet the independent condition in Proposition \ref{prop:imperceptiblity-requirements}. Each private key is then used for the mark module to modify the next token distribution. We maintain a vector database $[(Enc(\rvy^1), \rvk^1), ..., (Enc(\rvy^N), \rvk^N)]$ to store the semantic embedding $Enc(\rvy^n)$ of each output as keys and the corresponding private key vector $\rvk^n$ as values. For each candidate text $\hat \rvy$ during detection, regardless of whether it is watermarked, the most plausible private key vector $\hat\rvk$ is retrieved based on semantic similarity:
\vspace{-4pt}
\begin{equation*}
    \hat\rvk = \rvk^{n^*}\text{, where } n^*= \arg\max_n \text{sim} (Enc(\rvy^n), Enc(\hat\rvy))
\end{equation*}

\vspace{-10pt}
The restored key vector is then provided for the mark module to calculate the statistics. Given that the most similar text to the candidate is certainly its own even under attacks, this approach guarantees the accurate key restoration if the candidate is watermarked.

As a key module, WaterPool is able to integrate with most of mark modules. In this study, we compose WaterPool with mark modules of three well known watermarking techniques including KGW\citep{kirchenbauerWatermarkLargeLanguage2023}, ITS\citep{kuditipudiRobustDistortionfreeWatermarks2023} and EXP\citep{kuditipudiRobustDistortionfreeWatermarks2023}. These improved watermarks are referred to as KGW-Pool, ITS-Pool and EXP-Pool, respectively. The details of implementation and pseudo codes can be found in Appendix \ref{app:pseudo-code}.

Building on Proposition \ref{prop:imperceptiblity-requirements}, the imperceptibility of WaterPool is readily established with the mark modules of EXP and ITS satisfying the unbiased condition \footnote{The mark module of KGW does not satisfy the unbiased condition. Therefore, WaterPool can only enhance its imperceptibility performance, but not achieve optimal imperceptibility.}.
\vspace{-5pt}
\begin{proposition}
\label{prop:waterpool-imperceptibility}
    Both ITS-Pool and EXP-Pool are imperceptible.
\end{proposition}

\vspace{-5pt}
Furthermore, the efficacy of WaterPool can also be assured based on the efficacy of combined mark module. This can be formalized as, 
\vspace{-2pt}
\begin{proposition}
\label{prop:waterpool-efficacy}
    The statistical difference in WaterPool is bounded from below, as expressed by:
    \begin{equation*}
        \E[S(\hat\rvk_i,\rvy_i)|H_1]-\E[S(\hat\rvk_i,\rvy_i)|H_0]\geq p_{recall}\cdot \phi (\vp^i)
    \end{equation*}
    ,where $\phi (\vp)$ is watermarking potentials of the corresponding mark module depending on the probability vector $\vp^i$ for the $i$-th token distribution $P_M(\cdot|\rvx,\rvy_{<i})$.
\end{proposition}

\vspace{-5pt}
This proposition indicates that WaterPool can effectively leverage the power of mark modules (i.e. the lower bound of its statistical difference with golden private key restoration), slightly modulated by the recall performance $p_{recall}$ of WaterPool's retriever. We will empirically demonstrate that the recall performance is near-optimal even in large scale database employing relatively weak retrievers in the following experiments.


\vspace{-5pt}
\subsection{Difference from Retrieval Watermark}
\label{sec:diff-from-retrieval-watermark}
\vspace{-5pt}

\citet{krishnaParaphrasingEvadesDetectors2023} proposed a retrieval watermark to distinguish watermarked texts by using semantic retrieval. They store every generated output $o$ in a vector database $D$, and utilize $\max_{o\in D}\text{sim}(o, o^\text{candidate})$ as score of the candidate being watermarked.

WaterPool fundamentally differs from retrieval watermark. WaterPool’s efficacy and robustness rely primarily on the statistical difference guaranteed by the mark module. In specific, $\E[S(\rvk_i,\rvy_i)]$ is designed to be high if $\rvy_i$ is sampled from $\rvk_i$-induced modified distribution and low if $\rvk_i$ and $\rvy_i$ are independent. Therefore, the retriever of WaterPool only needs to retrieve the correct private key for watermarked candidates, without concern for retrieval results of non-watermarked candidates, since all keys stored in the database are independent of all non-watermarked texts.

In contrast, retrieval watermark directly uses similarity as score. For efficacy, it should ensure high similarity scores for watermarked texts and low similarity scores for non-watermarked ones. However, due to the dense semantic space of human language, texts stored in the database often have similar non-watermarked neighbors (\textit{semantic collisions} in \citet{krishnaParaphrasingEvadesDetectors2023}), and hence reducing efficacy. This issue will become more severe as the number of non-watermarked samples or the size of the vector database increases. We conduct an experiment to empirically demonstrate this by adding responses from other models to the same prompt as non-watermarked samples (see Appendix \ref{app:retrieval-watermark-problem}). In this scenario, the performance of retrieval watermark is over 40\% worse than other watermarking methods.

\vspace{-2pt}
\section{Experiments}
\vspace{-2pt}
\paragraph{Datasets.} Following previous works \citep{kirchenbauerWatermarkLargeLanguage2023,kirchenbauerReliabilityWatermarksLarge2023}, we include the C4 dataset for open-ended generation and "Explain Like I'm Five" (ELI5) \citep{fan-etal-2019-eli5} for long-form question answering in our experiments. We sample about 3000 texts from each dataset as prompts.
\vspace{-5pt}
\paragraph{Metrics.} We generate 20 outputs for each prompt, a total of about 60000 samples for evaluation. For both efficacy and robustness, we report true positive rate at 1\% false positive rate (TPR@FPR=1\%). We also include the ROC-AUC metrics in Appendix \ref{app:additional-experiment}. Outputs of the original LLM are considered as non-watermarked. To comprehensively evaluate the robustness of watermarking techniques, we include three different kinds of attacks, namely Lexical-Attack, Dipper-Attack and Translation-Attack. Lexical-attack is a baseline attack by randomly add/delete/replace a small portion of texts. Dipper-attack is a paraphrasing model proposed by \citet{krishnaParaphrasingEvadesDetectors2023}. Translation-attack represents roundtrip-translation, translating texts to another language and then translating them back. For imperceptibility, we split the criteria into two aspects: (1) the distribution bias within each output (2) the independence among different outputs. The former can be evaluated with perplexity while the latter can be roughly evaluated with n-gram distinction \citep{kirchenbauerReliabilityWatermarksLarge2023}. Specifically, we consider the distinction across all outputs (Glob-Distinct-$N$) and within outputs in response to one single prompt (Group-Distinct-$N$).


\paragraph{Implementation details.} 
We conduct experiments on OPT-1.3b and OPT-6.7b following \citet{krishnaParaphrasingEvadesDetectors2023}, using multinomial sampling to generate outputs within the range of [50, 70] tokens. We include several typical methods as baselines, including Gamma\citep{huUnbiasedWatermarkLarge2023}, Delta\citep{huUnbiasedWatermarkLarge2023}, Unigram\citep{zhaoProvableRobustWatermarking2023}, KGW\citep{kirchenbauerWatermarkLargeLanguage2023}, EXP\citep{kuditipudiRobustDistortionfreeWatermarks2023} and ITS\citep{kuditipudiRobustDistortionfreeWatermarks2023}. We use recommended hyper-parameter settings for all baselines in their original paper.
For implementation of mark modules in different WaterPool (i.e. KGW-Pool, ITS-Pool, EXP-Pool), we use identical hyper-parameter settings as the original one. We use a 128 dimension sentence embedding model \citep{nussbaum2024nomic} as the retriever. Other details are presented in Appendix \ref{app:experiment-details}.

\vspace{-5pt}
\subsection{Main Results}
\vspace{-5pt}

{
\begin{table}[!tb]
\caption{Imperceptibility of different watermarking methods on OPT-1.3B. $\Delta$ is the difference between watermarked texts and non-watermarked texts. The best and second-best results are highlighted in \textbf{bold} and \underline{underline}.}
\label{tab:imperceptibility-opt1.3b}
\renewcommand{\arraystretch}{1.1}
\setlength{\tabcolsep}{1.1pt} 
\scriptsize
\centering
\makebox[\textwidth][c]{
\resizebox{\linewidth}{!}{
\begin{tabular}{l|rr|rr|rr|rr|rr} 
    \toprule
    \multicolumn{1}{c}{~} & \multicolumn{2}{c}{Glob-distinct2} & \multicolumn{2}{c}{Glob-distinct3} & \multicolumn{2}{c}{Group-distinct2} & \multicolumn{2}{c}{Group-distinct3} & \multicolumn{2}{c}{ppl} \\
    \midrule
        \multicolumn{1}{c}{~} & \multicolumn{1}{c}{value$\uparrow$} & \multicolumn{1}{c}{$\Delta$$\uparrow$} & \multicolumn{1}{c}{value$\uparrow$} & \multicolumn{1}{c}{$\Delta$$\uparrow$} & \multicolumn{1}{c}{value$\uparrow$} & \multicolumn{1}{c}{$\Delta$$\uparrow$} & \multicolumn{1}{c}{value$\uparrow$} & \multicolumn{1}{c}{$\Delta$$\uparrow$} & \multicolumn{1}{c}{value$\downarrow$} & \multicolumn{1}{c}{$\Delta$$\downarrow$}\\
    \midrule
    \multicolumn{1}{c}{~} & \multicolumn{10}{c}{Open Text Generation} \\
    \midrule
Non-watermark  &  38.8$_{\pm 0.0}$ & 0.0$_{\pm 0.0}$ & \underline{76.2$_{\pm 0.0}$} & \underline{0.0$_{\pm 0.0}$} & 86.2$_{\pm 0.0}$ & 0.0$_{\pm 0.0}$ & \underline{96.2$_{\pm 0.0}$} & \underline{0.0$_{\pm 0.0}$} & \underline{7.8$_{\pm 0.0}$} & \underline{0.0$_{\pm 0.0}$} \\
Gamma  &  38.8$_{\pm 0.0}$ & 0.0$_{\pm 0.0}$ & 76.2$_{\pm 0.0}$ & -0.0$_{\pm 0.0}$ & 86.2$_{\pm 0.0}$ & -0.0$_{\pm 0.0}$ & 96.2$_{\pm 0.0}$ & -0.0$_{\pm 0.0}$ & 7.8$_{\pm 0.0}$ & 0.0$_{\pm 0.0}$ \\
Delta  &  38.8$_{\pm 0.0}$ & -0.0$_{\pm 0.0}$ & 76.2$_{\pm 0.0}$ & -0.0$_{\pm 0.1}$ & 86.2$_{\pm 0.0}$ & -0.0$_{\pm 0.1}$ & 96.2$_{\pm 0.0}$ & -0.0$_{\pm 0.0}$ & 7.8$_{\pm 0.0}$ & 0.0$_{\pm 0.0}$ \\
Unigram  &  33.4$_{\pm 1.9}$ & -5.3$_{\pm 1.9}$ & 69.9$_{\pm 2.8}$ & -6.3$_{\pm 2.7}$ & 82.6$_{\pm 2.4}$ & -3.6$_{\pm 2.4}$ & 95.1$_{\pm 0.6}$ & -1.1$_{\pm 0.6}$ & 9.9$_{\pm 0.6}$ & 2.2$_{\pm 0.6}$ \\
KGW  &  36.7$_{\pm 0.1}$ & -2.0$_{\pm 0.2}$ & 73.6$_{\pm 0.1}$ & -2.7$_{\pm 0.1}$ & 85.5$_{\pm 0.1}$ & -0.7$_{\pm 0.1}$ & 95.9$_{\pm 0.0}$ & -0.3$_{\pm 0.1}$ & 9.6$_{\pm 0.0}$ & 1.8$_{\pm 0.0}$ \\
KGW-Pool  &  \textbf{40.5$_{\pm 0.2}$} & \textbf{1.8$_{\pm 0.2}$} & \textbf{78.7$_{\pm 0.2}$} & \textbf{2.4$_{\pm 0.2}$} & \textbf{87.3$_{\pm 0.1}$} & \textbf{1.1$_{\pm 0.2}$} & \textbf{96.7$_{\pm 0.0}$} & \textbf{0.5$_{\pm 0.1}$} & 9.9$_{\pm 0.0}$ & 2.1$_{\pm 0.0}$ \\
EXP  &  30.2$_{\pm 0.0}$ & -8.5$_{\pm 0.0}$ & 59.2$_{\pm 0.0}$ & -17.0$_{\pm 0.0}$ & 73.4$_{\pm 0.1}$ & -12.8$_{\pm 0.0}$ & 82.2$_{\pm 0.1}$ & -14.0$_{\pm 0.0}$ & 7.8$_{\pm 0.0}$ & 0.0$_{\pm 0.0}$ \\
EXP-Pool  &  38.7$_{\pm 0.0}$ & -0.0$_{\pm 0.0}$ & 76.2$_{\pm 0.0}$ & -0.0$_{\pm 0.0}$ & 86.2$_{\pm 0.0}$ & -0.0$_{\pm 0.1}$ & 96.2$_{\pm 0.0}$ & -0.0$_{\pm 0.0}$ & 7.8$_{\pm 0.0}$ & 0.0$_{\pm 0.0}$ \\
ITS  &  34.4$_{\pm 0.7}$ & -4.4$_{\pm 0.7}$ & 66.4$_{\pm 1.5}$ & -9.8$_{\pm 1.5}$ & 75.2$_{\pm 1.8}$ & -11.0$_{\pm 1.8}$ & 83.7$_{\pm 2.1}$ & -12.6$_{\pm 2.1}$ & \textbf{7.5$_{\pm 0.0}$} & \textbf{-0.3$_{\pm 0.0}$} \\
ITS-Pool  &  \underline{38.8$_{\pm 0.0}$} & \underline{0.0$_{\pm 0.0}$} & 76.2$_{\pm 0.0}$ & -0.0$_{\pm 0.0}$ & \underline{86.2$_{\pm 0.0}$} & \underline{0.0$_{\pm 0.0}$} & 96.2$_{\pm 0.0}$ & -0.0$_{\pm 0.0}$ & 7.8$_{\pm 0.0}$ & 0.0$_{\pm 0.0}$ \\

    \midrule
    \multicolumn{1}{c}{~} & \multicolumn{10}{c}{Long-Form Question Answering} \\
    \midrule

Non-watermark  &  31.5$_{\pm 0.0}$ & 0.0$_{\pm 0.0}$ & 70.0$_{\pm 0.0}$ & 0.0$_{\pm 0.0}$ & 86.7$_{\pm 0.0}$ & 0.0$_{\pm 0.0}$ & 97.0$_{\pm 0.0}$ & 0.0$_{\pm 0.0}$ & \underline{9.5$_{\pm 0.0}$} & \underline{0.0$_{\pm 0.0}$} \\
Gamma  &  31.5$_{\pm 0.0}$ & 0.0$_{\pm 0.0}$ & 70.0$_{\pm 0.0}$ & 0.0$_{\pm 0.1}$ & 86.7$_{\pm 0.0}$ & 0.0$_{\pm 0.0}$ & 97.0$_{\pm 0.0}$ & 0.0$_{\pm 0.0}$ & 9.5$_{\pm 0.0}$ & 0.0$_{\pm 0.0}$ \\
Delta  &  31.5$_{\pm 0.0}$ & 0.0$_{\pm 0.0}$ & 70.0$_{\pm 0.0}$ & -0.0$_{\pm 0.1}$ & 86.8$_{\pm 0.0}$ & 0.0$_{\pm 0.0}$ & 97.0$_{\pm 0.0}$ & 0.0$_{\pm 0.0}$ & 9.5$_{\pm 0.0}$ & 0.0$_{\pm 0.0}$ \\
Unigram  &  26.4$_{\pm 2.0}$ & -5.1$_{\pm 2.0}$ & 62.0$_{\pm 2.9}$ & -7.9$_{\pm 2.9}$ & 81.2$_{\pm 2.5}$ & -5.5$_{\pm 2.5}$ & 94.8$_{\pm 0.6}$ & -2.2$_{\pm 0.6}$ & 10.9$_{\pm 1.4}$ & 1.4$_{\pm 1.4}$ \\
KGW  &  29.5$_{\pm 0.2}$ & -2.0$_{\pm 0.2}$ & 66.2$_{\pm 0.3}$ & -3.8$_{\pm 0.2}$ & 85.4$_{\pm 0.1}$ & -1.3$_{\pm 0.1}$ & 96.4$_{\pm 0.0}$ & -0.7$_{\pm 0.1}$ & 11.6$_{\pm 0.1}$ & 2.1$_{\pm 0.1}$ \\
KGW-Pool  &  \textbf{32.9$_{\pm 0.2}$} & \textbf{1.4$_{\pm 0.2}$} & \textbf{71.6$_{\pm 0.3}$} & \textbf{1.6$_{\pm 0.3}$} & 83.8$_{\pm 0.3}$ & -2.9$_{\pm 0.3}$ & 94.3$_{\pm 0.2}$ & -2.8$_{\pm 0.2}$ & 10.9$_{\pm 0.1}$ & 1.4$_{\pm 0.1}$ \\
EXP  &  22.6$_{\pm 0.3}$ & -8.9$_{\pm 0.3}$ & 50.0$_{\pm 0.7}$ & -20.0$_{\pm 0.7}$ & 75.3$_{\pm 1.1}$ & -11.4$_{\pm 1.1}$ & 85.1$_{\pm 1.4}$ & -11.9$_{\pm 1.4}$ & 9.6$_{\pm 0.0}$ & 0.1$_{\pm 0.0}$ \\
EXP-Pool  &  \underline{31.5$_{\pm 0.0}$} & \underline{0.0$_{\pm 0.0}$} & \underline{70.0$_{\pm 0.0}$} & \underline{0.0$_{\pm 0.0}$} & \textbf{86.8$_{\pm 0.0}$} & \textbf{0.1$_{\pm 0.0}$} & \textbf{97.1$_{\pm 0.0}$} & \textbf{0.0$_{\pm 0.0}$} & 9.5$_{\pm 0.0}$ & 0.0$_{\pm 0.0}$ \\
ITS  &  27.8$_{\pm 0.6}$ & -3.7$_{\pm 0.6}$ & 60.7$_{\pm 1.4}$ & -9.3$_{\pm 1.4}$ & 76.1$_{\pm 1.9}$ & -10.6$_{\pm 1.9}$ & 84.7$_{\pm 2.3}$ & -12.3$_{\pm 2.3}$ & \textbf{9.1$_{\pm 0.0}$} & \textbf{-0.4$_{\pm 0.0}$} \\
ITS-Pool  &  31.5$_{\pm 0.0}$ & -0.0$_{\pm 0.0}$ & 69.9$_{\pm 0.0}$ & -0.0$_{\pm 0.0}$ & \underline{86.8$_{\pm 0.0}$} & \underline{0.1$_{\pm 0.0}$} & \underline{97.0$_{\pm 0.0}$} & \underline{0.0$_{\pm 0.0}$} & 9.8$_{\pm 0.0}$ & 0.3$_{\pm 0.0}$ \\
\bottomrule
\end{tabular}
}
}
\end{table}
}

{
\begin{table}[!htb]
\caption{Efficacy and robustness of different watermarking methods on OPT-1.3B evaluated with TPR@FPR=1\%. $\Delta$ is the performance boost brought by WaterPool. The best and second-best results are highlighted in \textbf{bold} and \underline{underline}.}
\label{tab:robustness-opt1.3b}
\renewcommand{\arraystretch}{1}
\setlength{\tabcolsep}{2pt} 
\scriptsize
\centering
\makebox[\textwidth][c]{
\resizebox{1\linewidth}{!}{
\begin{tabular}{l|rr|rr|rr|rr} 
    \toprule
    \multicolumn{1}{c}{~} & \multicolumn{2}{c}{w/o Attack} & \multicolumn{2}{c}{Lexical-Attack} & \multicolumn{2}{c}{Dipper-Attack} & \multicolumn{2}{c}{Translation-Attack} \\
    \midrule
    \multicolumn{1}{c}{~} & \multicolumn{1}{c}{value$\uparrow$} & \multicolumn{1}{c}{$\Delta$} & \multicolumn{1}{c}{value$\uparrow$} & \multicolumn{1}{c}{$\Delta$} & \multicolumn{1}{c}{value$\uparrow$} & \multicolumn{1}{c}{$\Delta$} & \multicolumn{1}{c}{value$\uparrow$} & \multicolumn{1}{c}{$\Delta$} \\
    \midrule
    \multicolumn{1}{c}{~} & \multicolumn{8}{c}{Open Text Generation} \\
    \midrule
Gamma  &  96.94$_{\pm 0.05}$ & \multicolumn{1}{c|}{-} & 17.91$_{\pm 0.28}$ & \multicolumn{1}{c|}{-} & 2.26$_{\pm 0.04}$ & \multicolumn{1}{c|}{-} & 3.25$_{\pm 0.10}$ & \multicolumn{1}{c}{-} \\
Delta  &  75.37$_{\pm 0.34}$ & \multicolumn{1}{c|}{-} & 8.58$_{\pm 0.27}$ & \multicolumn{1}{c|}{-} & 2.07$_{\pm 0.08}$ & \multicolumn{1}{c|}{-} & 2.91$_{\pm 0.11}$ & \multicolumn{1}{c}{-} \\
\midrule
Unigram  &  93.98$_{\pm 1.32}$ & \multicolumn{1}{c|}{-} & 89.69$_{\pm 3.67}$ & \multicolumn{1}{c|}{-} & 19.99$_{\pm 7.44}$ & \multicolumn{1}{c|}{-} & 35.35$_{\pm 8.74}$ & \multicolumn{1}{c}{-} \\
KGW  &  \textbf{98.43$_{\pm 0.08}$} & \multicolumn{1}{c|}{-} & 88.88$_{\pm 0.13}$ & \multicolumn{1}{c|}{-} & 15.05$_{\pm 0.32}$ & \multicolumn{1}{c|}{-} & 29.53$_{\pm 0.23}$ & \multicolumn{1}{c}{-} \\
KGW-Pool  &  98.29$_{\pm 0.01}$ & -0.15$_{\pm 0.08}$ & \underline{95.29$_{\pm 1.04}$} & 6.41$_{\pm 1.17}$ & \underline{24.62$_{\pm 2.01}$} & 9.57$_{\pm 2.32}$ & \underline{42.26$_{\pm 1.70}$} & 12.73$_{\pm 1.73}$ \\
\midrule
EXP  &  97.19$_{\pm 0.08}$ & \multicolumn{1}{c|}{-} & 93.48$_{\pm 0.09}$ & \multicolumn{1}{c|}{-} & 18.32$_{\pm 0.36}$ & \multicolumn{1}{c|}{-} & 31.14$_{\pm 0.33}$ & \multicolumn{1}{c}{-} \\
EXP-Pool  &  \underline{98.43$_{\pm 0.01}$} & 1.24$_{\pm 0.09}$ & \textbf{96.67$_{\pm 0.07}$} & 3.19$_{\pm 0.15}$ & \textbf{26.17$_{\pm 0.86}$} & 7.85$_{\pm 0.62}$ & \textbf{51.41$_{\pm 0.42}$} & 20.27$_{\pm 0.64}$ \\
\midrule
ITS  &  73.43$_{\pm 0.10}$ & \multicolumn{1}{c|}{-} & 26.23$_{\pm 0.19}$ & \multicolumn{1}{c|}{-} & 2.16$_{\pm 0.08}$ & \multicolumn{1}{c|}{-} & 3.56$_{\pm 0.08}$ & \multicolumn{1}{c}{-} \\
ITS-Pool  &  92.56$_{\pm 0.14}$ & 19.12$_{\pm 0.11}$ & 68.50$_{\pm 0.46}$ & 42.27$_{\pm 0.40}$ & 4.05$_{\pm 0.15}$ & 1.89$_{\pm 0.11}$ & 10.83$_{\pm 0.31}$ & 7.27$_{\pm 0.39}$ \\

    \midrule
    \multicolumn{1}{c}{~} & \multicolumn{8}{c}{Long-Form Question Answering} \\
    \midrule

Gamma  &  98.68$_{\pm 0.05}$ & \multicolumn{1}{c|}{-} & 21.20$_{\pm 0.36}$ & \multicolumn{1}{c|}{-} & 2.31$_{\pm 0.05}$ & \multicolumn{1}{c|}{-} & 5.18$_{\pm 0.21}$ & \multicolumn{1}{c}{-} \\
Delta  &  90.19$_{\pm 0.13}$ & \multicolumn{1}{c|}{-} & 12.17$_{\pm 0.10}$ & \multicolumn{1}{c|}{-} & 2.21$_{\pm 0.08}$ & \multicolumn{1}{c|}{-} & 4.96$_{\pm 0.02}$ & \multicolumn{1}{c}{-} \\
\midrule
Unigram  &  96.93$_{\pm 1.99}$ & \multicolumn{1}{c|}{-} & 92.47$_{\pm 3.77}$ & \multicolumn{1}{c|}{-} & 26.38$_{\pm 5.88}$ & \multicolumn{1}{c|}{-} & 43.17$_{\pm 6.98}$ & \multicolumn{1}{c}{-} \\
KGW  &  \underline{99.51$_{\pm 0.01}$} & \multicolumn{1}{c|}{-} & 94.12$_{\pm 0.06}$ & \multicolumn{1}{c|}{-} & 19.21$_{\pm 0.17}$ & \multicolumn{1}{c|}{-} & 46.62$_{\pm 0.43}$ & \multicolumn{1}{c}{-} \\
KGW-Pool  &  99.51$_{\pm 0.00}$ & -0.01$_{\pm 0.01}$ & \underline{97.97$_{\pm 0.04}$} & 3.85$_{\pm 0.08}$ & \underline{29.92$_{\pm 1.04}$} & 10.71$_{\pm 1.20}$ & 50.14$_{\pm 0.24}$ & 3.52$_{\pm 0.37}$ \\
\midrule
EXP  &  99.17$_{\pm 0.06}$ & \multicolumn{1}{c|}{-} & 97.56$_{\pm 0.08}$ & \multicolumn{1}{c|}{-} & 27.92$_{\pm 0.56}$ & \multicolumn{1}{c|}{-} & \underline{54.99$_{\pm 0.37}$} & \multicolumn{1}{c}{-} \\
EXP-Pool  &  \textbf{99.56$_{\pm 0.04}$} & 0.40$_{\pm 0.08}$ & \textbf{98.81$_{\pm 0.07}$} & 1.25$_{\pm 0.05}$ & \textbf{36.24$_{\pm 0.82}$} & 8.32$_{\pm 1.37}$ & \textbf{72.61$_{\pm 0.29}$} & 17.62$_{\pm 0.56}$ \\
\midrule
ITS  &  86.40$_{\pm 0.51}$ & \multicolumn{1}{c|}{-} & 38.23$_{\pm 0.46}$ & \multicolumn{1}{c|}{-} & 3.02$_{\pm 0.19}$ & \multicolumn{1}{c|}{-} & 8.19$_{\pm 0.17}$ & \multicolumn{1}{c}{-} \\
ITS-Pool  &  97.56$_{\pm 0.05}$ & 11.16$_{\pm 0.46}$ & 81.73$_{\pm 0.34}$ & 43.51$_{\pm 0.21}$ & 6.25$_{\pm 0.15}$ & 3.23$_{\pm 0.05}$ & 24.26$_{\pm 0.43}$ & 16.07$_{\pm 0.27}$ \\
    
\bottomrule
\end{tabular}
}
}
\vspace{-15pt}
\end{table}
}

Results of WaterPool on OPT-1.3b are presented in Table \ref{tab:imperceptibility-opt1.3b} and \ref{tab:robustness-opt1.3b}. Results on OPT-6.7b are provided in Appendix \ref{app:additional-experiment}. Each experiment is repeated for three times.
Most of watermarking techniques achieve great performance in efficacy (more than 90\% TPR@FPR=1\%).
Consistent with our theoretical analysis in Section \ref{sec:trade-off-under-key-centered-framework}, Unigram, EXP, and ITS restrict the key space $\Xi$, thereby compromising the independence of keys across different generations. This is evidenced by the weak performance on Distinct-$N$ metric. On the other hand, Gamma and Delta achieve optimal imperceptibility 
but at the expense of efficacy and robustness.
WaterPool addresses these trade-offs to a significant extent. The results show that WaterPool markedly improves the original watermarking techniques across both tasks and all three properties. It elevates the imperceptibility of original watermarks to near-optimal levels, as indicated by the minimal difference from non-watermarked texts. Additionally, it consistently enhances the efficacy and robustness of the original techniques, as demonstrated by the substantial improvements in the TPR@FPR=1\% metric.

\subsection{Real-world Challenges for WaterPool}

As stated in Section \ref{sec:diff-from-retrieval-watermark}, WaterPool only requires to retrieve golden private key if it exists. This assertion holds intuitively, as a watermarked text, even under attacks, should remain semantically closer to the original watermarked text stored in the database than other texts. Otherwise, it should not be considered a modified version of the original watermarked text. To empirically demonstrate this, we conduct various experiments to show WaterPool's stability under real-world challenges.

\paragraph{Performance with diverse negative samples.} We test WaterPool's performance with different non-watermarked text distributions, including human-written outputs (\textit{Human}, 3K samples) and other non-watermarked models' outputs\footnote{These outputs are identical to non-watermarked texts used in Section \ref{sec:diff-from-retrieval-watermark}. Please refer to Appendix \ref{app:retrieval-watermark-problem}.} (\textit{Other Models}, 1.8M samples). Specifically, we utilize six models: Gemma-2b, Gemma-7b, Llama2-7b, Llama2-13b, Vicuna-7b, Vicuna-13b to generate responses for the same prompts. Results shown in Table \ref{tab:different-neg-samples} indicate that all WaterPool methods exhibit stable performance regardless of the number and types of negative samples, which aligns with the theoretical analysis above.

\vspace{-5pt}
\paragraph{Scalability with large vector databases.} 
We also conducted experiments to scale up the vector database size, simulating real-world scenarios. Specifically, we augment the database with noisy entries by incorporating 50-token fragments from the C4 dataset, constructing a noisy database of about 100 million items. During detection, if a noisy item is retrieved, a random key will be used for statistic calculation, hence affecting WaterPool's performance. This setting is extremely challenging since the noisy database shares the similar distribution with watermarked texts, both from C4 dataset. Results presented in Figure \ref{fig:scaling-noisy-database} and Appendix \ref{app:scaling-noisy-database}, demonstrate that performance of WaterPool remains robust as database size increases exponentially. It further indicates the feasibility of our proposed WaterPool in real world scenarios.

\vspace{-5pt}
\subsection{Analysis of Space and Time}
\vspace{-5pt}

The price for breaking trade-offs among imperceptibility, efficacy, and robustness with WaterPool lies in its space usage. The space complexity of WaterPool is $O(N)$, growing linearly with the number of generations. However, this is still practical in real-world scenarios. In our experiments, we use sentence embeddings in form of 128-dimensional bfloat16 arrays and a private key of one int32 number\footnote{In practice, we can only sample a number as seed to initialize a pseudo-random number generator, thereby generating the whole private key vector $\rvk\in\R^L$ for all tokens.}. Given ChatGPT's monthly visits are about 2B per month\footnote{\url{https://explodingtopics.com/blog/chatgpt-users}}, it takes 260 bytes of space per item, resulting in about 520 GB of storage per month, which is certainly manageable nowadays.

Regarding time complexity, the main cost of WaterPool comes from retrieval. It requires less than 0.001 sec per item when conducting retrieval on a database of 100M items with 10 NVIDIA GeForce RTX 3090 GPUs. It is worth noting that retrieval is a well-established field with various acceleration techniques available, which we have not explored in this paper. Additionally, watermark detection is not typically a time-intensive application. After all, we think WaterPool is highly practical in real world deployments.

\begin{figure}[!tp]
\centering
{
\begin{minipage}[t]{0.62\textwidth}
\renewcommand{\arraystretch}{1.2}
\setlength{\tabcolsep}{0.5pt} 
\centering
\captionof{table}{TPR@FPR=1\% of WaterPool with different non-watermarked texts in form of \lstinline|(C4 Result/LFQA Result)|. The first column lists watermarking methods, and the second column shows non-watermarked text sources. WaterPool remains stable across different non-watermarked texts.}
\label{tab:different-neg-samples}
\makebox[\textwidth][c]{
\resizebox{\linewidth}{!}{
    \centering
    \begin{tabular}{cc|cccc}
    \toprule
     ~ & ~ & ~~~w/o Attack~~~~ & ~Lexical-Attack~ & ~Dipper-Attack~ & Translation-Attack\\
    \midrule
\multirow{3}{*}{EXP-Pool}  &  Original  &  98.43 / 99.56 & 96.67 / 98.81 & 26.17 / 36.24 & 51.41 / 72.61 \\
 &  Human  &  98.35 / 99.66 & 96.50 / 99.04 & 24.49 / 39.70 & 50.23 / 75.00 \\
 &  Other Models  &  98.45 / 99.61 & 96.72 / 98.88 & 25.65 / 36.63 & 51.60 / 72.51 \\
 \midrule
\multirow{3}{*}{KGW-Pool}  &  Original  &  98.29 / 99.51 & 95.29 / 97.97 & 24.62 / 29.92 & 42.26 / 50.14 \\
 &  Human  &  98.96 / 99.98 & 96.27 / 99.56 & 27.23 / 46.80 & 46.22 / 69.24 \\
 &  Other Models  &  98.96 / 99.71 & 96.21 / 98.03 & 26.33 / 30.52 & 45.42 / 51.66 \\
 \midrule
\multirow{3}{*}{ITS-Pool}  &  Original  &  92.56 / 97.56 & 68.50 / 81.73 & 4.05 / 6.25 & 10.83 / 24.26 \\
 &  Human  &  93.33 / 97.44 & 71.18 / 80.77 & 4.61 / 5.82 & 12.01 / 22.95 \\
 &  Other Models  &  92.73 / 97.65 & 69.27 / 82.22 & 4.15 / 6.62 & 11.17 / 24.92 \\
\bottomrule
    \end{tabular}
}
}
\end{minipage}
}%
\hfill
{
\begin{minipage}[t]{0.33\textwidth}
\centering
\caption{TPR@FPR=1\% of WaterPool with different database size on open-ended generation.}
\label{fig:scaling-noisy-database}
\vspace{-5pt}
\includegraphics[width=\linewidth]{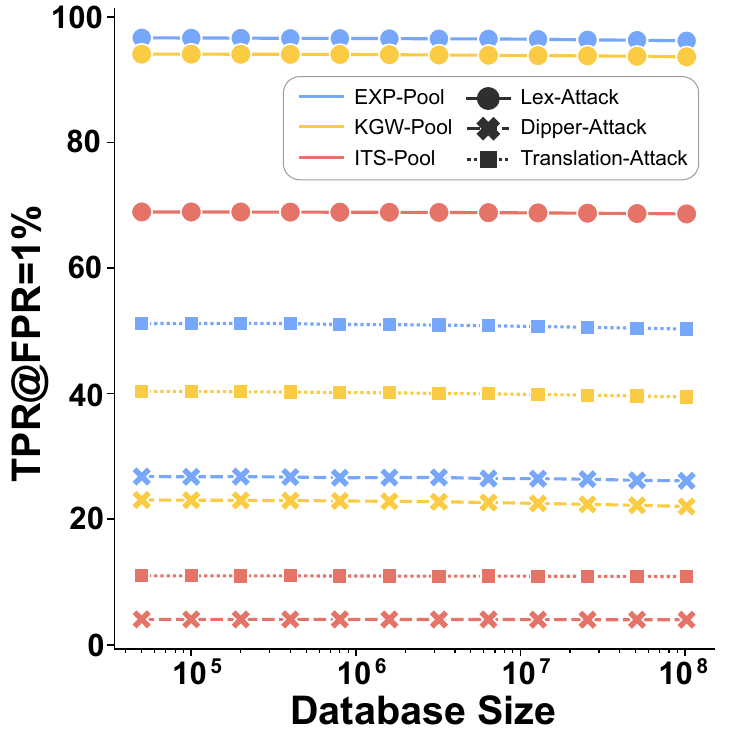}
\end{minipage}
}
\vspace{-15pt}
\end{figure}

\vspace{-5pt}
\section{Related Work}
\vspace{-5pt}

Significant progress has been made in the field of text watermarking, particularly focusing on private watermarking methods integrated into the text generation process of large language models (LLMs), which require no additional training. KGW \citep{kirchenbauerWatermarkLargeLanguage2023} is the first watermarking in the era of LLM, using context windows and hashing as private keys to increase probabilities of a random partition of vocabulary by adding their logits. Building on KGW, many different methods made some improvements, such as proposing different hash functions \citep{kirchenbauerReliabilityWatermarksLarge2023,houSemStampSemanticWatermark2023,liuSemanticInvariantRobust2023,renRobustSemanticsbasedWatermark2023}, heuristic partition strategies \citep{liImprovingGenerationQuality2023,chenXMarkLosslessWatermarking2023}, embedding multi-bit messages \citep{wangCodableTextWatermarking2023}, and robust hypothesis testing techniques \citep{fernandezThreeBricksConsolidate2023}.
Unigram \citep{zhaoProvableRobustWatermarking2023} stands out for fixing the private key for better efficacy and robustness. On the contrary, Gamma and Delta\citep{huUnbiasedWatermarkLarge2023} focus on imperceptibility with improved probability modification mechanisms.
Another notable work by \citet{christUndetectableWatermarksLanguage2023}  examines LLM generation from a bit-level perspective, modifying the distribution through inverse transform sampling. 
\citet{kuditipudiRobustDistortionfreeWatermarks2023} propose two similar techniques, ITS and EXP, which apply inverse transform sampling and Gumbel-max sampling at the token level. Despite these advancements, the trade-off among imperceptibility, efficacy, and robustness has been widely recognized and remains unresolved. A concurrent work \citep{giboulotWaterMaxBreakingLLM2024} also tries to break the trade-off by resampling until significant watermarking signals are observed. Although this approach maintains imperceptibility in a single turn, it significantly alters the $N$-shot output distribution.

\vspace{-5pt}
\section{Conclusion and Future Work}
\vspace{-5pt}

In this paper, we focus on the trade-off challenges among imperceptibility, efficacy and robustness in language model watermarking. Through a key-centered scheme, we have identified that these trade-offs arise from the conflict between the scale of the key sampling space during generation and the complexity of key restoration during detection. We propose WaterPool, a novel key module utilizing semantic search to alleviate this conflict. WaterPool can integrate with most existing watermarks, significantly enhancing their performance across all three properties and thus mitigating the trade-off issue.
Although WaterPool is simple and easy to deploy, it does not completely eliminate the trade-offs. It serves as an effective key module towards this goal, but the mark module is also crucial. We introduce the unbiased condition and statistical difference to outline the requirements for mark modules. Specific mark module designs are not thoroughly discussed in this work, as it is not our main focus. Future research could explore better mark module, further advancing watermarking techniques towards resolving the trade-offs among imperceptibility, efficacy, and robustness.



\bibliographystyle{plainnat}
\bibliography{ref,watermark}

\newpage
\appendix

\clearpage
\section{Algorithms of WaterPool}
\label{app:pseudo-code}

In this section, we present the pseudo codes for EXP-Pool, ITS-Pool and KGW-Pool. We highlight the codes WaterPool invokes in with triangle comments ($\triangleright$), while the other parts remain unchanged from the original watermarking techniques. The invocation of WaterPool is little, demonstrating the ease of its integration with other existing watermarking methods. The specific implementation of modification function $F$ and per-token statistic $S$ used in the pseudo codes can be found in Table \ref{tab:prior-watermark-design-mark}.

\algnewcommand{\LineComment}[1]{\State {\footnotesize\textcolor{gray}{/*  \texttt{#1}  */} }}

\vspace{-5pt}
\begin{algorithm}[!htb]
    \caption{EXP-Pool Generation}
    \label{alg:exp-pool-gen}
    \textbf{Params}: language model $M$, max output length $L$, reweight function of EXP $F_{exp}(\cdot,\cdot)$, embedding model $Enc(\cdot)$.\\
    \textbf{Input}: $N$ rounds queries $\{\rvx^n\}^N_{n=1}$.\\
    \textbf{Output}: $N$ rounds outputs $\{\rvy^n\}^N_{n=1}$, vector database $D$.
    \begin{algorithmic}[1]
        \State $D \gets \{\}$\Comment{Initialize vector database}
        \LineComment{Multi-round queries}
        \For{$n\in\{1,...,N\}$}
            \State Input current round prompt $\rvx^n$
            \State $\rvy^n\gets \text{empty string}$
            \State $\rvk^n\sim\mathcal U(\R^L)$ \Comment{Sample key}
            \LineComment{Auto-regressive generation}
            \For{$i\in1,...,L$}
                \State $P_i\gets P_M(\cdot|\rvx^n,\rvy^n_{<i})$
                \State $\hat P_i\gets F_{exp}(\rvk^n_i,P_i)$
                \State $\rvy^n_i\sim \hat P_i$
            \EndFor
            \State $D\gets D \cup \{(Enc(\rvy^n),\rvk^n)\}$ \Comment{Store key for detection}
            \State Output current round generation $\rvy^n$
        \EndFor
    \end{algorithmic} 
\end{algorithm}
\vspace{-3pt}

\begin{algorithm}[!htb]
    \caption{EXP-Pool Detection}
    \label{alg:exp-pool-detect}
    \textbf{Params}: vector database $D=\{(Enc(\rvy^n),\rvk^n)\}_n^N$, embedding model $Enc(\cdot)$, permutation resample times $T$, edit penalty $\eta$, per-token statistic of EXP $S_{exp}(\cdot,\cdot)$.\\
    \textbf{Input}: candidate text $\hat\rvy$\\
    \textbf{Output}: $p$-value of being watermarked $\hat p$
    \begin{algorithmic}[1]
        \LineComment{Aggregation of per-token statistic $s_i$ with edit distance trick}
        \Procedure{$d_{\textnormal{edit}}$}{$\rvk, \rvy$}:
            \If{$\text{len}( \rvk) = 0$}
                \State\Return{$-\eta\cdot \text{len}( \rvy)$}
            \ElsIf{$\text{len}( \rvy) = 0$}
                \State\Return{$-\eta\cdot \text{len}( \rvk)$}
            \Else
                \State $s_i\gets S_{exp}( \rvk_1, \rvy_2)$
                \State\Return{$\max\{
                d_{\text{edit}}( \rvk_{2:}, \rvy_{2:})+s_i,
                d_{\text{edit}}( \rvk_{2:}, \rvy)-\eta,
                d_{\text{edit}}( \rvk, \rvy_{2:})-\eta
                \}$}
            \EndIf
        \EndProcedure
        \Statex
        \State $n^*\gets\argmax_n\text{sim}(Enc(\rvy^n),Enc(\hat\rvy))$ \Comment{Retrieve key from vector database}
        \State $\hat\rvk\gets\rvk^{n^*}$
        \State $\hat V\gets d_{edit}(\hat\rvk,\hat\rvy)$
        \LineComment{Permutation test}
        \For{$t\in 1,...,T$}
            \State $\rvk^t\sim\mathcal{U}(\R^L)$
            \State $V^t\gets d_{edit}(\rvk^t,\hat\rvy)$
        \EndFor
        \LineComment{Calculate $p$-value}
        \State $\hat p\gets\frac{1}{T+1}\left(1+\sum_t \1_{\hat V>V^t}\right)$
    \end{algorithmic}
\end{algorithm}

\begin{algorithm}[!htb]
    \caption{ITS-Pool Generation}
    \label{alg:its-pool-gen}
    \textbf{Params}: language model $M$, max output length $L$, reweight function of ITS $F_{its}(\cdot,\cdot)$, embedding model $Enc(\cdot)$.\\
    \textbf{Input}: $N$ rounds queries $\{\rvx^n\}^N_{n=1}$.\\
    \textbf{Output}: $N$ rounds outputs $\{\rvy^n\}^N_{n=1}$, vector database $D$.
    \begin{algorithmic}[1]
        \State $D \gets \{\}$\Comment{Initialize vector database}
        \LineComment{Multi-round queries}
        \For{$n\in\{1,...,N\}$}
            \State Input current round prompt $\rvx^n$
            \State $\rvy^n\gets \text{empty string}$
            \State $\rvk^n\sim\mathcal U(\R^L)$ \Comment{Sample key}
            \LineComment{Auto-regressive generation}
            \For{$i\in1,...,L$}
                \State $P_i\gets P_M(\cdot|\rvx^n,\rvy^n_{<i})$
                \State $\hat P_i\gets F_{its}(\rvk^n_i,P_i)$
                \State $\rvy^n_i\sim \hat P_i$
            \EndFor
            \State $D\gets D \cup \{(Enc(\rvy^n),\rvk^n)\}$ \Comment{Store key for detection}
            \State Output current round generation $\rvy^n$
        \EndFor
    \end{algorithmic} 
\end{algorithm}

\begin{algorithm}[!htb]
    \caption{ITS-Pool Detection}
    \label{alg:its-pool-detect}
    \textbf{Params}: vector database $D=\{(Enc(\rvy^n),\rvk^n)\}_n^N$, embedding model $Enc(\cdot)$, permutation resample times $T$, edit penalty $\eta$, per-token statistic of ITS $S_{its}(\cdot,\cdot)$.\\
    \textbf{Input}: candidate text $\hat\rvy$\\
    \textbf{Output}: $p$-value of being watermarked $\hat p$
    \begin{algorithmic}[1]
        \LineComment{Aggregation of per-token statistic $s_i$ with edit distance trick}
        \Procedure{$d_{\textnormal{edit}}$}{$\rvk, \rvy$}:
            \If{$\text{len}( \rvk) = 0$}
                \State\Return{$-\eta\cdot \text{len}( \rvy)$}
            \ElsIf{$\text{len}( \rvy) = 0$}
                \State\Return{$-\eta\cdot \text{len}( \rvk)$}
            \Else
                \State $s_i\gets S_{its}( \rvk_1, \rvy_2)$
                \State\Return{$\max\{
                d_{\text{edit}}( \rvk_{2:}, \rvy_{2:})+s_i,
                d_{\text{edit}}( \rvk_{2:}, \rvy)-\eta,
                d_{\text{edit}}( \rvk, \rvy_{2:})-\eta
                \}$}
            \EndIf
        \EndProcedure
        \Statex
        \State $n^*\gets\argmax_n\text{sim}(Enc(\rvy^n),Enc(\hat\rvy))$ \Comment{Retrieve key from vector database}
        \State $\hat\rvk\gets\rvk^{n^*}$
        \State $\hat V\gets d_{edit}(\hat\rvk,\hat\rvy)$
        \LineComment{Permutation test}
        \For{$t\in 1,...,T$}
            \State $\rvk^t\sim\mathcal{U}(\R^L)$
            \State $V^t\gets d_{edit}(\rvk^t,\hat\rvy)$
        \EndFor
        \LineComment{Calculate $p$-value}
        \State $\hat p\gets\frac{1}{T+1}\left(1+\sum_t \1_{\hat V>V^t}\right)$
    \end{algorithmic}
\end{algorithm}

\begin{algorithm}[!htb]
    \caption{KGW-Pool Generation}
    \label{alg:kgw-pool-gen}
    \textbf{Params}: language model $M$, max output length $L$, reweight function of KGW $F_{kgw}(\cdot,\cdot)$, embedding model $Enc(\cdot)$.\\
    \textbf{Input}: $N$ rounds queries $\{\rvx^n\}^N_{n=1}$.\\
    \textbf{Output}: $N$ rounds outputs $\{\rvy^n\}^N_{n=1}$, vector database $D$.
    \begin{algorithmic}[1]
        \State $D \gets \{\}$\Comment{Initialize vector database}
        \LineComment{Multi-round queries}
        \For{$n\in\{1,...,N\}$}
            \State Input current round prompt $\rvx^n$
            \State $\rvy^n\gets \text{empty string}$
            \State $k^n\sim\mathcal U(\R)$ \Comment{Sample key}
            \LineComment{Auto-regressive generation}
            \For{$i\in1,...,L$}
                \State $P_i\gets P_M(\cdot|\rvx^n,\rvy^n_{<i})$
                \State $\hat P_i\gets F_{kgw}(k^n,P_i)$
                \State $\rvy^n_i\sim \hat P_i$
            \EndFor
            \State $D\gets D + \{(Enc(\rvy^n),k^n)\}$ \Comment{Store key for detection}
            \State Output current round generation $\rvy^n$
        \EndFor
    \end{algorithmic} 
\end{algorithm}

\begin{algorithm}[!htb]
    \caption{KGW-Pool Detection}
    \label{alg:kgw-pool-detect}
    \textbf{Params}: vector database $D=\{(Enc(\rvy^n),k^n)\}$, embedding model $Enc(\cdot)$, per-token statistic of KGW $S_{kgw}(\cdot,\cdot)$ \\
    \textbf{Input}: candidate text $\hat\rvy$\\
    \textbf{Output}: $p$-value of being watermarked $\hat p$
    \begin{algorithmic}[1]
        \State $n^*\gets\argmax_n\text{sim}(Enc(\rvy^n),Enc(\hat\rvy))$ \Comment{Retrieve key from vector database}
        \State $\hat k \gets k^{n^*}$
        \LineComment{Aggregation of per-token statistic $s_i$ via summation}
        \For{$i\in 1,...,\text{len}(\hat\rvy)$}
            \State $s_i\gets S_{kgw}(\hat k, \hat\rvy_i)$
        \EndFor
        \LineComment{Calculation of $z$-score}
        \State $z\gets \sum_i s_i$
        \LineComment{Calculation of $p$-value}
        \State $p\gets 1-\Phi(z)$
    \end{algorithmic}
\end{algorithm}

\clearpage
\section{Theoretical Proofs}




\subsection{Proof of Proposition \ref{prop:imperceptiblity-requirements}}

\begin{proof}
Recall that the imperceptibility is defined as 
\begin{equation*}
    \prod_{i,n}P_M(\rvy_i^n|\rvx^n,\rvy_{<i}^n)=\E_{\rvk^1,...,\rvk^N}[\prod_{i,n} F(\rvk_i^n,P_i^n)(\rvy_i^n)],\, \forall \rvx^n,\rvy^n\in\Sigma^*
\end{equation*}
Given that (1) $\rvk^1,...,\rvk^N\overset{i.i.d}{\sim}\mathcal{U}(\R^L)$; (2) $P_M(\cdot |\rvx^n,\rvy_{<i}^n) = \E_{\rvk_i\sim\mathcal U(\R)}[F(\rvk_i,P_i)]$. We have,
\begin{align*}
    \text{RHS} &= \E_{\rvk^1\sim\mathcal{U}(\R^L)}...\E_{\rvk^N\sim\mathcal{U}(\R^L)}[\prod_{i,n} F(\rvk_i^n,P_i^n)(\rvy_i^n)]\\
    &=\prod_n\E_{\rvk^n\sim\mathcal{U}(\R^L)}[\prod_iF(\rvk_i^n,P_i^n)(\rvy_i^n)]\\
    &=\prod_n(\E_{\rvk_1^n\sim\mathcal U(\R)}...\E_{\rvk_L^n\sim\mathcal U(\R)}[\prod_iF(\rvk_i^n,P_i^n)(\rvy_i^n)])\\
    &=\prod_n(\prod_i\E_{\rvk_i^n\sim\mathcal U(\R)}[F(\rvk_i^n,P_i^n)(\rvy_i^n)])\\
    &=\prod_{n,i}P_M(\rvy_i^n|\rvx^n,\rvy^n_{<i})=\text{LHS}
\end{align*}
Where the first four steps utilize the mutual independence among private keys, while the last step is based on the second condition.
\end{proof}


\subsection{Proof of Proposition \ref{prop:waterpool-imperceptibility} (EXP-Pool's Imperceptibility)}
\label{app:exp-imperceptibility}

We first recall the distribution modification process of EXP-Pool (i.e. $F_{exp}(\rvk_i,P_i)$). Given a private key $\rvk_i$ as seed, a standard Gumbel vector $\rvg_{\rvk_i}\in\R^{|\Sigma|}$ is sampled. EXP-Pool then conducts a Gumbel-max sampling on the next token distribution $P_i(\cdot):=P_M(\cdot|\rvx^n,\rvy_{<i}^n)$ via $\rvg_{\rvk_i}$ to sample an output token $t^*$. The degenerate distribution of $t^*$ is then returned.

\begin{lemma}
    The mark module of EXP-Pool satisfies the unbiased condition, i.e.
    \begin{equation*}
        P_M(\cdot |\rvx^n,\rvy_{<i}^n) = \E_{\rvk_i\sim\mathcal U(\R)}[F_{exp}(\rvk_i,P_i)]
    \end{equation*}
\end{lemma}
\begin{proof}
    For simplicity, we denote $P_M(t |\rvx^n,\rvy_{<i}^n)$ as $P_i(t)$ and $g_t$ as the Gumbel variable in $\rvg_{\rvk_i}$ corresponding to the token $t$. Since 
    The lemma holds if and only if for any token $t\in\Sigma$,
    \begin{equation*}
        P_M(t |\rvx^n,\rvy_{<i}^n) = \E_{\rvg_{\rvk_i,j}\overset{i.i.d}{\sim}\text{Gumbel}(0,1)}[\1_{\log P_i(t)+g_t\geq \log P_i(t')+g_{t'},\forall t'\in\Sigma}]
    \end{equation*}
    This equation follows as
    \begin{align*}
        \text{RHS}&=P(\log P_i(t)+g_t\geq \log P_i(t')+g_{t'},\forall t'\in\Sigma) &&\\
        &=P(\exp(-\exp(-g_{t'}))\leq \exp(-\exp(-g_t))^{P_i(t')/P_i(t)},\forall t'\in\Sigma) &&\\
        &=\int_0^1 P(u_{t'}\leq u_t^{P_i(t')/P_i(t)},\forall t'\in\Sigma|u_t)p(u_t)du_t\\
        &=\int_0^1 \prod_{t'\in\Sigma}P(u_{t'}\leq u_t^{P_i(t')/P_i(t)}|u_t)p(u_t)du_t\\
        &=\int_0^1 u_t^{\sum_{t'} P_i(t')/P_i(t)}du_t &&\\
        &=P_i(t) = \text{LHS} &&\\
    \end{align*}
    ,where $u_t:=\exp(-\exp(-g_{t}))$. We have $u_t\sim\mathcal U(0,1)$, since $g_t\sim\text{Gumbel}(0,1)$.
\end{proof}

Given the unbiased condition guaranteed by this lemma and the independent condition guaranteed by WaterPool, we can immediately have the imperceptibility of EXP-Pool according to Proposition \ref{prop:imperceptiblity-requirements}.

\subsection{Proof of Proposition \ref{prop:waterpool-imperceptibility} (ITS-Pool's Imperceptibility)}
\label{app:its-imperceptibility}

We first recall the distribution modification process of ITS-Pool (i.e. $F_{its}(\rvk_i,\rvy_i)$). Given a private key $\rvk_i$ as the random seed, a random permutation $\pi_{\rvk_i}:\Sigma\rightarrow\Sigma$ and a uniform variable $\ru_{\rvk_i}\sim\mathcal U([0,1])$ are sampled. ITS conducts an inverse transform sampling on the permuted distribution $P^\text{perm}$ via $\ru$. The sampled token $t$ is transformed back to $t^*$ via inverse permutation $\pi_{\rvk_i}^{-1}$. The degenerate distribution of $t^*$ is then returned.

\begin{lemma}
    The mark module of ITS satisfies the unbiased condition, i.e.
    \begin{equation*}
        P_M(\cdot |\rvx^n,\rvy_{<i}^n) = \E_{\rvk_i\sim\mathcal U(\R)}[F_{its}(\rvk_i,P_i)]
    \end{equation*}
\end{lemma}

\begin{proof}
The lemma follows if that given any permutation $\pi$ and any output token $t^*$,
\begin{equation*}
    P_M(\pi(t^*)|\rvx^n,\rvy_{<i}^n) = \E_{\ru_{\rvk_i}\sim\mathcal U([0,1])}[\1_{\pi(t^*)\text{ is sampled via inverse transform sampling}}]
\end{equation*}
This equation certainly holds because of the definition of inverse transform sampling, i.e.
\begin{align*}
    \text{RHS} =& P(\ru_{\rvk_i}\in
    [
    P_M(\{t':\text{ord}(t')<\text{ord}\circ\pi(t^*)\}|\rvx^n,\rvy_{<i}^n),\\
    &~~~~~~~~~~~~~~~~P_M(\{t':\text{ord}(t')\leq\text{ord}\circ\pi(t^*)\}|\rvx^n,\rvy_{<i}^n)
    \,]\,)\\
    =& P_M(\pi(t^*)|\rvx^n,\rvy_{<i}^n) = \text{LHS}
\end{align*}
,where $ord:\Sigma\rightarrow|\Sigma|$ is a function maps each token to its order in vocabulary.
\end{proof}

Given the unbiased condition guaranteed by this Lemma and the independent condition guaranteed by WaterPool, we immediately have the imperceptibility of ITS-Pool according to Proposition \ref{prop:imperceptiblity-requirements}.

\subsection{Proof of Proposition \ref{prop:waterpool-efficacy} (EXP-Pool's Efficacy) }
{
\newcommand{\key}{\ensuremath{{\rvk_i}}}
\newcommand{\token}{\ensuremath{{\rvy_i}}}
\newcommand{\prob}{\ensuremath{{P_i(\rvy_i)}}}
\newcommand{\gumbel}{\ensuremath{{g_{\rvy_i}}}}
\newcommand{\othertoken}{\ensuremath{{t}}}
\newcommand{\othergumbel}{\ensuremath{{g_{t}}}}
\newcommand{\otherprob}{\ensuremath{{P_i({t})}}}
\newcommand{\prefix}{\ensuremath{{\rvy_{<i}}}}
\newcommand{\ord}{\ensuremath{\text{ord}}}
\newcommand{\probFunc}[1]{\ensuremath{{P_i(#1)}}}

We first recall the mark module of EXP-Pool (\textit{gumbel-sample} in Table \ref{tab:prior-watermark-design-mark}). 
During generation, the modification function $F_{exp}(\rvk_i,P_i)$ takes a private key $\rvk_i$ as seed to sample a standard Gumbel vector. A Gumbel-max sampling is then conduct to sample an output token, of which the degenerate distribution is returned.
During detection, the mark module takes in a restored private key $\hat\key$ to generate a restored Gumbel vector $\rvg_{\hat\key}\in\R^{|\Sigma|},\rvg_{\hat\key,j}\sim\text{Gumbel}(0,1)$. The per-token statistic $S_{exp}(\hat\rvk_i,\rvy_i)=-\exp(-\rvg_{\hat\key,\ord(\rvy_i)})$ is then calculated, where $\ord: \Sigma\rightarrow \{1,...,|\Sigma|\}$ is a function maps each token to the its order in the vocabulary. 

For simplicity, we denote $g_t$ as the Gumbel variable $\rvg_{\rvk_i,\ord(t)}$ corresponding to the token $t$.

We begin by proving that given a candidate token $\token$, the expectation of per-token statistic only relies on the original distribution $P_i(\rvy_i)$ if $\rvy_i$ is sampled from the modified distribution seeded by $\rvk_i$. It can be formalized as the following lemma.
\begin{lemma}
\label{lem:lem-exp-pool-efficacy}
    Given a prefix $\rvy_{<i}$ and a token $\rvy_i$, for a private key $\key\sim\mathcal U(\R)$, if $\rvy_i$ is sampled from $F_{exp}(\rvk_i,P_M(\cdot|\rvy_{<i}))$, then $\E_{\key\sim\R}[S_{exp}(\rvk_i,\rvy_i)|\token,\prefix] = -\probFunc{\token}$.
\end{lemma}
\begin{proof}
    The randomness of $\key$ affects the statistic via the Gumbel vector $\rvg_\key=[g_t]_{t\in\Sigma}$.
    For simplicity, we only consider the randomness of $\rvg_\key$ instead of $\key$.

    We first calculate the cumulative distribution function of $S_{exp}(\key,\token)|\token,\prefix$.
    \begin{align*}
        &P(S_{exp}(\key,\token)\leq v|\token,\prefix)\\
        =& P(g_\token\leq-\log(-v)|\token,\prefix)\\
        \overset{(1)}{=}& P\left(\left.\bigcap_{\othertoken\in\Sigma} 
        \othergumbel+\log\otherprob
        \leq   
        \gumbel + \log\prob
        \leq 
        -\log(-v) + \log\prob
        \right|\token,\prefix
        \right)\\
        =& 
        \prod_{t\in\Sigma}
        P\left(\left.
        g_t\leq \log\frac{\otherprob}{-v\prob}
        \right|\token,\prefix
        \right)\\
        \overset{(2)}{=}& \prod_{t\in\Sigma} \exp(-v\frac{\otherprob}{\prob})\\
        =& \exp(-v/\prob)
    \end{align*}
    , where the equation $(1)$ follows from the definition of Gumbel max sampling; $(2)$ follows from $g_t\sim\text{Gumbel}(0,1) $. Therefore, $-S_{exp}(\key,\token)|\token,\prefix\sim Exp(1/\prob)$. The lemma follows immediately by calculating the expectation.
\end{proof}

On the contrary, if $\token$ is not sampled from modified distribution seeded by $\key$,
\begin{equation*}
    P(S_{exp}(\key,\token)\leq v|\token,\prefix)=P(g_\token\leq-\log(-v))=\exp(v)
\end{equation*}
, from which we have $-S_{exp}(\key,\token)|\token,\prefix\sim Exp(1)$ and thus $\E_{\key\sim\mathcal U(\R)}[S_{exp}(\rvk_i,\rvy_i)|\token,\prefix]=-1$. 

Eventually, we can guaranteed the statistical difference of EXP-Pool given the prefix $\prefix$, 
\begin{align*}
    &\E[S_{exp}(\rvk_i,\rvy_i)|\prefix, H_1]-\E[S_{exp}(\rvk_i,\rvy_i)|\prefix, H_0]\\
    =&
    \E_{\token,\key}[S_{exp}(\hat\key,\token)|\prefix, H_1,\hat\key=\key]\cdot p_{recall}\\
    &+\E_{\token,\key}[S_{exp}(\hat\key,\token)|\prefix, H_1,\hat\key\neq\key]\cdot (1-p_{recall})+ 1\\
    =&\E_\token[-\prob|\prefix]\cdot p_{recall} - (1-p_{recall})+1\\
    =&\sum_{\token\in\Sigma}(1-\prob)\prob\cdot p_{recall}:=\phi_{exp}(\vp^i)\cdot p_{recall}
\end{align*}
, where $\phi_{exp}(\vp^i)$ is only relevant to probability vector $\vp^i$ of $P_M(\cdot|\prefix)$, representing watermarking potentials at this step, and $p_{recall}$ is the recall performance of the retriever in EXP-Pool.

}

\subsection{Proof of Proposition \ref{prop:waterpool-efficacy} (ITS-Pool's Efficacy) }
{
    

\newcommand{\key}{\ensuremath{{\rvk_i}}}
\newcommand{\token}{\ensuremath{{\rvy_i}}}
\newcommand{\prefix}{\ensuremath{{\rvy_{<i}}}}
\newcommand{\ord}{\ensuremath{\text{ord}}}
\newcommand{\permToken}{\ensuremath{\pi(\token)}}
\newcommand{\probFunc}[1]{\ensuremath{{P_i(#1)}}}

We first recall the mark module of ITS-Pool (\textit{inverse-sample} in Table \ref{tab:prior-watermark-design-mark}). During generation, the modification function $F_{its}(\key,P_i)$ takes a private key $\key$ as seed to sample a standard uniform variable and a random permutation. An inverse transform sampling is then conduct on the permuted distribution to sample an output token, of which the degenerate distribution is returned.
During detection, the mark module takes in a restored private key $\hat\key$ to restore a standard uniform variable $\ru_{\hat\key}\sim\mathcal{U}(0,1)$ and a random permutation $\pi_{\hat\key}:\Sigma\rightarrow \Sigma$. The per-token statistic $S_{its}(\hat\key,\token)=(\ru_{\hat\key}-\frac12)(\frac{\text{ord}(\pi_{\hat\rvk_i}(\rvy_i))-1}{|\Sigma|-1}-\frac12)$ is then calculated, where $\ord: \Sigma\rightarrow \{1,...,|\Sigma|\}$ is a function maps each token to the its order in the vocabulary. 


We begin by proving that given a candidate token $\token$, the expectation of per-token statistic only relies on the original distribution $P_i(\rvy_i)$ if $\rvy_i$ is sampled from the modified distribution seeded by $\rvk_i$ (i.e. the alternative hypothesis $H_1$). It can be formalized as the following lemma\footnote{It is based on Lemma B.1 in \citet{kuditipudiRobustDistortionfreeWatermarks2023}}.
\begin{lemma}
\label{lem:lem-its-pool-efficacy}
    Given a prefix $\rvy_{<i}$ and a token $\rvy_i$, for a private key $\key\sim\mathcal U(\R)$, if $\rvy_i$ is sampled from $F_{its}(\rvk_i,P_M(\cdot|\rvy_{<i}))$, then $\E_{\key\sim\mathcal U(\R)}[S_{its}(\rvk_i,\rvy_i)|\token,\prefix] = C_0\cdot (1-\probFunc{\token})$, where $C_0$ is a constant relevant to vocabulary size $|\Sigma|$.
\end{lemma}
\begin{proof}
    The randomness of $\key$ affects the statistic via the uniform variable $\ru_\key$ and the permutation $\pi_\key$. 
    For simplicity, we omit the subscript $\key$ and only consider the randomness of $\ru$ and $\pi$.
    
    We show that $\pi|\token,\prefix$ is a uniform random variable over the permutation space.
    \begin{align*}
        P(\pi|\token,\prefix)=\frac{P(\token|\pi,\prefix)P(\pi)}{P(\token|\prefix)}=P(\pi)
    \end{align*}
    , where the second equation follows from that permutation won't affect the inverse transform sampling (see Appendix\ref{app:its-imperceptibility})

    We also prove that $\ru|\pi,\token,\prefix$ is a uniform random variable. We define the interval of $\token$ given $\pi$ in inverse transform sampling during ITS generation,
    \begin{equation*}
    I(\token,\pi)=[\probFunc{\{t':\ord(t')<\ord\circ\pi(\token)\}},\probFunc{\{t':\ord(t')\leq\ord\circ\pi(\token)\}}]
    \end{equation*}
    It is evident that $|I(\token,\pi)=\probFunc{\token}|$. Then for any interval $I\subset [0,1]$ we have
    \begin{equation*}
        P(\ru\in I|\token,\pi,\prefix)=\frac{P(\token,\ru\in I|\pi,\prefix)}{P(\token|\pi,\prefix)}=\frac{|I\cap I(\token,\pi)|}{|I(\token,\pi)|}
    \end{equation*}
    So $\ru|\pi,\token,\prefix\sim\mathcal U(I(\token,\pi))$. Then we have,
    \begin{align*}
        \E[\ru|\token,\pi(\token),\prefix] &=\E\left[\left.\probFunc{\{t':\ord(t')<\ord\circ\pi(\token)\}}+\frac{|I(\token,\pi)|}{2}\right|\token,\pi(\token),\prefix\right]\\
        &=\frac{(\pi(\token)-1)}{|\Sigma|-1}\cdot (1-\probFunc{\token})+\frac{\probFunc{\token}}{2}\\
        &=\frac12+(\frac{(\pi(\token)-1)}{|\Sigma|-1}-\frac12)(1-\probFunc{\token})
    \end{align*}
    
    It is evident that $\E[\ru]=\frac12$ and $\E[\frac{(\pi(\token)-1)}{|\Sigma|-1}]=\frac12$, since they are both uniform standard variables. Therefore, $S_{its}(\rvk_i,\rvy_i)$ essentially calculates the covariance between $\ru$ and $\frac{\pi(\token)-1}{|\Sigma|-1}$, which is tractable as following,
    \begin{align*}
        \E_{\key}[S_{its}(\rvk_i,\rvy_i)|\token,\prefix]
        &=
        \text{Cov}\left(\left.\ru,\frac{\pi(\token)-1}{|\Sigma|-1}\right|\token,\prefix\right)\\
        &=(\ru-\frac12)
        (\frac{\pi(\token)-1}{|\Sigma|-1}-\frac12)
        \cdot 
        P(\ru,\pi(\token)|\token,\prefix)\\
        &=
        \E[\ru-\frac12|\token,\pi(\token),\prefix]
        \cdot 
        (\frac{\pi(\token)-1}{|\Sigma|-1}-\frac12)
        \cdot 
        P(\pi(\token)|\token,\prefix)\\
        &=
        (1-\probFunc{\token})
        \cdot
        (\frac{\pi(\token)-1}{|\Sigma|-1}-\frac12)^2
        \cdot 
        P(\pi(\token)|\token,\prefix)\\
        &=
        (1-\probFunc{\token})
        \cdot
        \text{Var}\left(\left.\frac{\pi(\token)-1}{|\Sigma|-1}\right|\token,\prefix\right)\\
        &=C_0\cdot (1-\probFunc{\token})
    \end{align*}
    ,where $C_0=\text{Var}\left(\left.\frac{\pi(\token)-1}{|\Sigma|-1}\right|\token,\prefix\right)$ is a constant since $\pi|\token,\prefix$ is uniform over the space of vocabulary permutation.
\end{proof}

On the contrary, if $\token$ is not sampled from modified distribution seeded by $\key$, $\E_\key[S_{its}(\rvk_i,\rvy_i)|\token,\prefix]=\text{Cov}(\ru,\frac{\pi(\token)-1}{|\Sigma|-1}|\token,\prefix)$ still holds. Now that $\key$ and $\token$ are independent, $\E_\key[S_{its}(\rvk_i,\rvy_i)|\token,\prefix]=0$ trivially. Therefore, $\E_{\token,\key}[S_{its}(\hat\key,\token)|\prefix,H_0]=0$ follows immediately.

Under the alternative hypothesis $H_1$, the lemma above provides that
\begin{align*}
    \E_{\token,\key}[S_{its}(\hat\key,\token)|\prefix,H_1] 
    &=\E_{\token,\key}[S_{its}(\hat\key,\token)|\prefix,H_1,\hat\key=\key]\cdot p_{recall}\\
    &~~~~+ \E_{\token,\key}[S_{its}(\hat\key,\token)|\prefix,H_1,\hat\key\neq\key]\cdot (1-p_{recall})\\
    &=\E_\token[C_0\cdot (1-\probFunc{\token})|\prefix]\cdot p_{recall}\\
    &=C_0\cdot p_{recall}\cdot \sum_{\token\in\Sigma}(1-\probFunc{\token})\probFunc{\token}
\end{align*}
, where $p_{recall}$ represents the recall performance of the retriever in ITS-Pool.

Finally, we can guarantee the statistical difference of ITS-Pool,
\begin{equation*}
    \E[S_{its}(\key,\token)|\prefix,H_1]-\E[S_{its}(\key,\token)|\prefix,H_0]=C_0\cdot\sum_{\token\in\Sigma}(1-\probFunc{\token})\probFunc{\token}\cdot p_{recall}:=\phi_{its}(\vp^i)\cdot p_{recall}
\end{equation*}
, where $\phi_{its}(\vp^i)$ is only relevant to probability vector $\vp^i$ of $P_M(\cdot|\prefix)$, representing watermarking potentials at this step.

}

\subsection{Proof of Proposition \ref{prop:waterpool-efficacy} (KGW-Pool's Efficacy)}
{
\newcommand{\key}{\ensuremath{{\rvk_i}}}
\newcommand{\token}{\ensuremath{{\rvy_i}}}
\newcommand{\prob}{\ensuremath{{P_i(\rvy_i)}}}
\newcommand{\othertoken}{\ensuremath{{t}}}
\newcommand{\otherprob}{\ensuremath{{P_i({t})}}}
\newcommand{\prefix}{\ensuremath{{\rvy_{<i}}}}
\newcommand{\ord}{\ensuremath{\text{ord}}}
\newcommand{\probFunc}[1]{\ensuremath{{P_i(#1)}}}
\newcommand{\greentokens}{\ensuremath{{\mathcal G}}}
\newcommand{\vocabsize}{\ensuremath{{|\Sigma|}}}

We first recall the mark module of KGW-Pool (\textit{logits-add} in Table \ref{tab:prior-watermark-design-mark}). 
During generation, the mark module will randomly sample a green list $\greentokens_{\key}$ of $\gamma|\Sigma|$ tokens from vocabulary, which is seeded by $\key$. Logits of these green tokens are increased by a constant $\delta$ to form the modified distribution $F_{kgw}(\key,\prob)$.
During detection, the mark module takes in a restored private key $\hat\key$, generates a green list $\greentokens_{\hat\key}$ seeded by $\hat\key$, and then calculates the per-token statistic $S_{kgw}(\hat\rvk_i, \rvy_i)=\frac{\1_{ \rvy_i\in \mathcal G_{\hat\rvk_i}} - \gamma}{\sqrt{\text{len}(\rvy)\gamma(1-\gamma)}}$.

For simplicity, we omit the subscript $\key$ in $\greentokens_{\rvk_i}$. We also denote the size of vocabulary and green list as $N=|\Sigma|$ and $N_G=\gamma|\Sigma|$ respectively.

Since the denominator is a constant under both hypotheses, we only need to focus on $S'(\rvk_i, \rvy_i):=\1_{ \rvy_i\in \mathcal G_{\rvk_i}}$. Under the alternative hypothesis, the expectation of $S'$ is essentially the probability of sampling a token from the green list during KGW-Pool generation. We show that the probability can be bounded from below, as formalized in the following lemma\footnote{It is based on Lemma F.1 in \citet{kirchenbauerWatermarkLargeLanguage2023}}.
\begin{lemma}
\label{lem:lem-exp-pool-efficacy}
    Given a prefix $\rvy_{<i}$, for a private key $\key\sim\mathcal U(\R)$ and a token $\token$, if $\token$ is sampled from $F_{kgw}(\rvk_i,P_M(\cdot|\rvy_{<i}))$, then $\E_{\key,\token}[S'(\rvk_i,\rvy_i)|\prefix] \geq C_1\cdot\textup{Spike}(\vp^i,\frac{(1-\gamma)(\alpha-1)}{1+(\alpha-1)\gamma})$, where $C_1$ is a constant and $\textup{Spike}(\vp^i,c)$ is the spike entropy defined in \citet{kirchenbauerWatermarkLargeLanguage2023}.
\end{lemma}
\begin{proof}
    Trivially, we have
    $\E_{\key,\token}[S'(\rvk_i,\rvy_i)|\prefix]=P(\token\in\greentokens, \token, \greentokens|\prefix)$. We consider the following process of sampling $\token$ and $\greentokens$. We first randomly choose a token $t$ as output token $\token$, and then randomly sample the remaining tokens to construct the green list. Therefore, the expected probability of a token from the green list being sampled can be written as,
    \begin{equation*}   \E_{\token\in\Sigma}\E_{\greentokens \,s.t.\token\in\greentokens}\frac{\alpha\prob}{\sum_{t\in\Sigma}\otherprob+\alpha\sum_{t\in\greentokens}\otherprob}
    \end{equation*}
    , where $\alpha=\exp(\delta)$.

    Define the inner expectation as $f_\token(\vp^i)$, where $\vp^i$ is the probability vector of $P_i(\cdot)$. Trivially, $f_\token(\vp^i)=f_\token(\Pi\vp^i)$ for any permutation $\Pi$ over the vocabulary except $\token$. Also, $f_\token$ is convex in $\vp^i_{-\token}$. Therefore, we have,
    \begin{align*}
        f_\token(\vp^i) &= \E_\Pi f_\token(\Pi\vp^i)\\
        &\overset{(1)}{\geq} f_\token(\E_\Pi\Pi\vp^i)\\
        &\overset{(2)}{\geq} \frac{\alpha\prob}{(1-\prob)(N-N_G)/(N-1)+\alpha(1-\prob)(N_G-1)/(N-1)+\alpha\prob}\\
        &=\prob\frac{\alpha N-\alpha}{N-N_G+\alpha N_G+(\alpha-1)(N-N_G)\prob-\alpha}\\
        &\geq \prob\frac{\alpha N}{N-N_G+\alpha N_G+(\alpha-1)(N-N_G)\prob}\\
        &=\frac{\alpha\prob}{(1-\gamma)+\alpha\gamma + (\alpha-1)(1-\gamma)\prob}
    \end{align*}
    , where $(1)$ follows from Jensen's inequality; $(2)$ follows from $\E_\Pi\Pi\vp^i_t=\frac{1-\prob}{N-1},\,\forall t\neq\token$. Then we have,
    \begin{align*}
    \E_{\key\token}[S'(\rvk_i,\rvy_i)|\prefix]&=P(\token\in\greentokens, \token, \greentokens|\prefix)\\
    &=N_G\cdot \E_{\token\in\Sigma}\E_{\greentokens \,s.t.\token\in\greentokens}\frac{\alpha\prob}{\sum_{t\in\Sigma}\otherprob+\alpha\sum_{t\in\greentokens}\otherprob}\\
    &=N_G\cdot \E_{\token\in\Sigma} f_\token(\vp^i)\\
    &\geq \frac{\gamma\alpha}{1+(\alpha-1)\gamma} \text{Spike}(\vp^i,\frac{(1-\gamma)(\alpha-1)}{1+(\alpha-1)\gamma})\\
    &:=C_1\cdot\text{Spike}(\vp^i,\frac{(1-\gamma)(\alpha-1)}{1+(\alpha-1)\gamma})
    \end{align*}
    , where $\textup{Spike}(\vp^i,c)=\sum_{t\in\Sigma}\frac{\otherprob}{1+c\otherprob}$. And the lower bound is strictly larger than $\gamma$.
\end{proof}

On the contrary, under the null hypothesis, trivially we have $\E[S'(\rvk_i,\rvy_i)|\prefix, H_0]=\gamma$, since $\rvk_i$ and $\rvy_i$ are independent. Combining all above, we eventually have,
\begin{align*}
    &\E[S(\hat\rvk_i,\rvy_i)|\prefix, H_1]-\E[S(\hat\rvk_i,\rvy_i)|\prefix, H_0]\\
    =& \E_{\key,\token}[S(\hat\rvk_i,\rvy_i)|\prefix, H_1, \hat\key=\key]\cdot p_{recall}\\
    &+ \E_{\key,\token}[S(\hat\rvk_i,\rvy_i)|\prefix, H_1, \hat\key\neq\key]\cdot (1-p_{recall}) - 0\\
    =& (\E_{\key,\token}[S'(\rvk_i,\rvy_i)|\prefix, H_1]-\gamma)/\sqrt{\text{len}(\rvy)\gamma(1-\gamma)}\cdot p_{recall} + 0 - 0\\
    \geq& (C_1\cdot \text{Spike}(\vp^i,\frac{(1-\gamma)(\alpha-1)}{1+(\alpha-1)\gamma})-\gamma )/\sqrt{\text{len}(\rvy)\gamma(1-\gamma)}\cdot p_{recall}\\
    :=&\phi_{kgw}(\vp^i)\cdot p_{recall}\\
\end{align*}
, where $\phi_{kgw}(\vp^i)$ is only relevant to probability vector $\vp^i$ of $P_M(\cdot|\prefix)$, representing watermarking potentials at this step, and $p_{recall}$ is the recall performance of the retriever in KGW-Pool.

}

\clearpage
\section{Experimental Details}
\label{app:experiment-details}
\paragraph{Datasets.} Following previous works \citep{kirchenbauerWatermarkLargeLanguage2023,kirchenbauerReliabilityWatermarksLarge2023}, we include two common used datasets for our experiments, the Colossal Common Crawl Cleaned corpus (C4) and "Explain Like I'm Five" (ELI5) \citep{fan-etal-2019-eli5}. We randomly select 3000 texts of length 50 from C4 as prompts for open-ended generation task. As for ELI5, we use the version curated by \citet{krishnaParaphrasingEvadesDetectors2023} including 2758 samples for long-form question answering. 

\paragraph{Attack.} To comprehensively evaluate the robustness of watermarking techniques, we include three different kinds of attacks, namely Lexical, Dipper and Translation. Lexical-attack is a baseline attack by randomly add/delete/replace a small portion of texts. We randomly modify 10\% tokens of the sentence. Dipper is a paraphrasing model proposed by \citet{krishnaParaphrasingEvadesDetectors2023}. We set the \lstinline|lex=40,div=40| following \citep{kirchenbauerReliabilityWatermarksLarge2023}. Translation-attack represents roundtrip-translation, which is a widely used paraphrasing method. Following \citet{kuditipudiRobustDistortionfreeWatermarks2023}, we translate texts to Russian and then translate them back to English.


\paragraph{Implementation details.} 
We conduct experiments on two models of different scales, OPT-1.3b and OPT-6.7b, following \citet{krishnaParaphrasingEvadesDetectors2023}. 
On both open-ended generation and long-form question answering, we conduct multi-nomial sampling to generate sequences within the range of [50, 70] tokens. 
We use a 128 dimension sentence embedding model \citep{nussbaum2024nomic} as the retriever in WaterPool.
As for implementation of mark modules in different WaterPool (i.e. KGW-Pool, ITS-Pool, EXP-Pool), we use identical hyper-parameter settings as the original watermarking technique. All baselines are reproduced based on source codes provided by original paper. For KGW, we set $\delta=2.0,\,\gamma=0.25$ as suggested in \citet{kirchenbauerWatermarkLargeLanguage2023}. For Unigram, we set $\delta=2.0,\,\gamma=0.5$ as suggested in \citet{zhaoProvableRobustWatermarking2023}. For Gamma and Delta, we use the context length of $5$ and search the perturbation strength $d$ over the set$\{0,0.1,...,1.0\}$ following \citet{christUndetectableWatermarksLanguage2023}. For EXP and ITS, we set the key length $n=80$, large enough to generate at most $70$ tokens in our main experiments. We set the edit-distance penalty $\gamma=0.0$ and $0.4$ respectively following \citet{kuditipudiRobustDistortionfreeWatermarks2023}. 
For KGW-Pool, we observed high variance of performance. It is because of the random partition in a sequence level (line 6 in Algorithm \ref{alg:kgw-pool-gen}). The logits-add mark module is very sensitive to this partition, which is similar to Unigram. To this end, we resample the key for three times and use the best one as watermarked output during generation in practice. This trick only leads to additional time complexity of generation and won't affect any other analysis in this paper. EXP, ITS, EXP-Pool and ITS-Pool all leverage permutation tests to calculate $p$-value (\lstinline|/* Permutation test */| in Algorithm \ref{alg:exp-pool-detect} and \ref{alg:its-pool-detect}). We conduct the permutations with $5000$ resamples. Following \citet{kuditipudiRobustDistortionfreeWatermarks2023}, we only pre-compute the permutation distribution once instead of recomputing it for each candidate text. This trick reduces the high time complexity of permutation tests and doesn't cause much performance degradation.


\section{Additional Experiments}
\label{app:additional-experiment}

\subsection{Performance of WaterPool with OPT-6.7b}
In this section, we present results of WaterPool on OPT-6.7B in Table \ref{tab:imperceptibility-opt6.7b}, \ref{tab:robustness-opt6.7b} and \ref{tab:robustness-opt6.7b-roc}.

\subsection{Problem with Retrieval Watermark}
\label{app:retrieval-watermark-problem}

In this section, we conduct an experiment to empirically demonstrate the statements made in Section \ref{sec:diff-from-retrieval-watermark}. We utilize eight models: Gemma-2b, Gemma-7b, Llama2-7b, Llama2-13b, Vicuna-7b, Vicuna-13b, OPT-1.3b, and OPT-6.7b, to generate outputs for both open-ended generation and long-form question answering tasks. In our experimental setup, one model is treated as the watermarked model, while outputs from the other models are considered non-watermarked texts. This setup reflects a common real-world scenario. We evaluate the robustness of retrieval watermarking under lexical attack, the weakest form of attack, in Table \ref{tab:diff-from-retrieval}. The results show that even under the weakest attack, the retrieval watermark experiences a significant performance degradation of more than 40\%, while most of other watermarking techniques still achieve high TPR@FPR=1\% (see Table \ref{tab:robustness-opt1.3b} and \ref{tab:robustness-opt6.7b}). This underscores the vulnerability of retrieval watermarking compared to other methods in real-world applications.
{
\begin{table}[htb]
\renewcommand{\arraystretch}{1.1}
\setlength{\tabcolsep}{2pt} 
\centering
\caption{TPR@FPR=1\% of retrieval watermark under lexical attacks. The model in the first column is the model being watermarked. The results are presented in form of \lstinline|(C4 Result/LFQA Result)|. Retrieval watermark is vulnerable even under the weakest lexical attacks.}
\label{tab:diff-from-retrieval}
\scriptsize
\makebox[\textwidth][c]{
\resizebox{\linewidth}{!}{
    \centering
    \begin{tabular}{c|ccccccccc}
    \toprule
    Watermarked Model & Vicuna-13b & Vicuna-7b & Llama2-13b & Llama2-7b & Gemma-2b & Gemma-7b & OPT-1.3b & OPT-6.7b & Avg\\
    \midrule
Vicuna-13b & ~~~~-~~~~ / ~~~~-~~~~ & 16.36 / ~~0.87 & 32.36 / 39.39 & 40.85 / 44.02 & 61.23 / 68.47 & 51.09 / 45.95 & 49.94 / 55.36 & 40.02 / 46.01 & 41.69 / 42.87 \\
Vicuna-7b & 14.62 / ~~0.76 & ~~~~-~~~~ / ~~~~-~~~~ & 39.24 / 40.00 & 41.41 / 42.84 & 61.16 / 68.19 & 53.00 / 45.65 & 49.63 / 54.25 & 39.73 / 45.01 & 42.69 / 42.38 \\
Llama2-13b & 23.36 / ~~9.16 & 31.65 / ~~9.81 & ~~~~-~~~~ / ~~~~-~~~~ & 44.56 / 44.00 & 66.95 / 66.52 & 58.67 / 46.48 & 53.12 / 50.77 & 43.22 / 41.62 & 45.93 / 38.34 \\
Llama2-7b & 29.19 / 10.77 & 31.63 / 10.37 & 41.54 / 41.94 & ~~~~-~~~~ / ~~~~-~~~~ & 68.24 / 66.40 & 60.53 / 47.64 & 53.10 / 53.99 & 43.35 / 45.15 & 46.80 / 39.47 \\
Gemma-2b & 47.32 / 21.65 & 48.20 / 21.88 & 63.43 / 64.71 & 65.86 / 67.15 & ~~~~-~~~~ / ~~~~-~~~~ & 67.78 / 65.23 & 63.05 / 69.16 & 54.13 / 60.29 & 58.54 / 52.87 \\
Gemma-7b & 37.95 / ~~8.93 & 41.03 / ~~9.56 & 55.49 / 42.78 & 58.77 / 45.74 & 68.70 / 63.75 & ~~~~-~~~~ / ~~~~-~~~~ & 58.30 / 54.14 & 48.65 / 44.80 & 52.70 / 38.53 \\
OPT-1.3b & 42.29 / 17.56 & 42.29 / 17.56 & 53.00 / 47.64 & 57.93 / 47.64 & 66.59 / 64.59 & 62.43 / 52.26 & ~~~~-~~~~ / ~~~~-~~~~ & 53.00 / 52.26 & 53.93 / 42.79 \\
OPT-6.7b & 30.73 / 12.83 & 30.73 / ~~8.93 & 47.83 / 41.87 & 47.83 / 41.87 & 62.37 / 59.81 & 57.79 / 41.87 & 52.99 / 55.68 & ~~~~-~~~~ / ~~~~-~~~~ & 47.18 / 37.55 \\
\bottomrule
    \end{tabular}
}
}
\end{table}
}

\subsection{Different Scales of Vector Database}
\label{app:scaling-noisy-database}

We report the numeric results corresponding to Figure \ref{fig:scaling-noisy-database} in Table \ref{tab:scaling-noisy-database}.

\begin{table}[htb]
\centering
\scriptsize
\caption{TPR@FPR=1\% of WaterPool with different database size on open-ended generation task under three different attacks, namely Lex, Dipper and Translation.}
\label{tab:scaling-noisy-database}
\makebox[\textwidth][c]{
\resizebox{\linewidth}{!}{
    \begin{tabular}{cc|cccccccccccc}
\toprule
\multicolumn{2}{c|}{Database Size}   & 50K & 100K & 200K & 400K & 800K & 1.6M & 3.2M & 6.4M & 12.8M & 25.6M & 51.2M & 102.4M \\
\midrule
\multirow{3}{*}{EXP-Pool} & Lex & 96.68 & 96.66 & 96.64 & 96.60 & 96.60 & 96.56 & 96.53 & 96.52 & 96.45 & 96.36 & 96.31 & 96.21 \\
~ & Dipper & 26.80 & 26.75 & 26.78 & 26.71 & 26.60 & 26.63 & 26.66 & 26.47 & 26.45 & 26.37 & 26.18 & 26.11 \\
~ & Translation & 51.18 & 51.15 & 51.22 & 51.16 & 51.02 & 51.00 & 50.93 & 50.81 & 50.69 & 50.56 & 50.43 & 50.29 \\
\midrule
\multirow{3}{*}{KGW-Pool} & Lex & 94.09 & 94.08 & 94.08 & 94.06 & 94.03 & 94.00 & 93.95 & 93.92 & 93.85 & 93.81 & 93.75 & 93.65 \\
~ & Dipper & 23.05 & 23.03 & 22.99 & 22.96 & 22.90 & 22.85 & 22.78 & 22.65 & 22.52 & 22.36 & 22.21 & 22.02 \\
~ & Translation & 40.37 & 40.34 & 40.30 & 40.25 & 40.19 & 40.12 & 40.05 & 39.96 & 39.87 & 39.75 & 39.62 & 39.46 \\
\midrule
\multirow{3}{*}{ITS-Pool} & Lex & 68.95 & 68.94 & 68.93 & 68.92 & 68.90 & 68.89 & 68.87 & 68.83 & 68.79 & 68.75 & 68.69 & 68.63 \\
~ & Dipper & 4.04 & 4.04 & 4.03 & 4.04 & 4.04 & 4.03 & 4.02 & 4.03 & 3.99 & 3.99 & 3.98 & 3.98 \\
~ & Translation & 10.97 & 10.97 & 10.96 & 10.96 & 10.95 & 10.94 & 10.93 & 10.93 & 10.91 & 10.90 & 10.88 & 10.87 \\    
\bottomrule
\end{tabular}
}
}

\end{table}

\subsection{Scaling Length of Outputs}

\begin{figure}
    \centering
    \includegraphics[width=\linewidth]{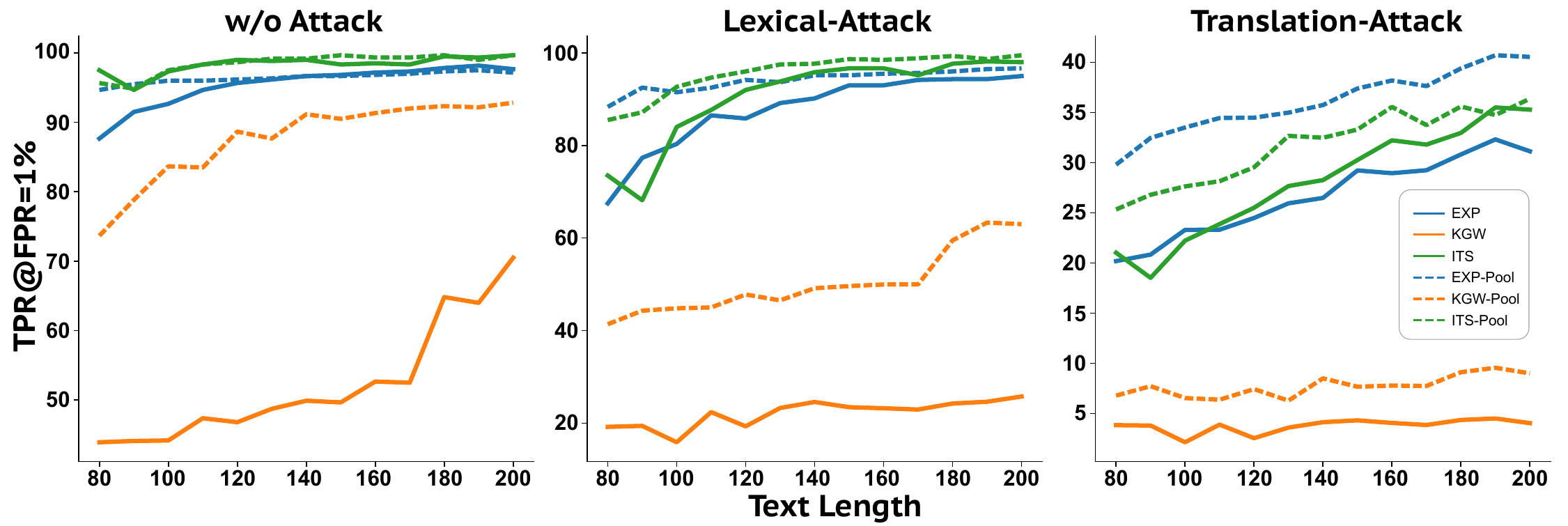}
    \caption{TPR@FPR=1\% of different watermarking techniques with the growths of text length. The same color indicates different methods sharing the same mark module. Solid lines represent original methods while dashed lines represent WaterPool methods.}
    \label{fig:scaling-length}
\end{figure}

We conduct another experiment to investigate the performance of WaterPool with growths of text length. We generate outputs of different lengths $T\in[80, 90, ... 200]$ and calculate the corresponding TPR@FPR=1\% metrics. The results are presented in Figure \ref{fig:scaling-length}. We observe a consistent increase in both the efficacy and robustness of WaterPool, aligning with the findings reported by \citet{kirchenbauerReliabilityWatermarksLarge2023}. Moreover, across all settings,
WaterPool consistently enhance the performance of original watermarking techniques, further underscoring its superior capabilities.

{
\begin{table}[!htb]
\caption{ROC-AUC of different watermarking methods on OPT-1.3B. $\Delta$ is the performance boost brought by WaterPool. The best and second-best results are highlighted in \textbf{bold} and \underline{underline}.}
\label{tab:robustness-opt1.3b-roc}
\renewcommand{\arraystretch}{0.9}
\setlength{\tabcolsep}{2pt} 
\scriptsize
\centering
\makebox[\textwidth][c]{
\resizebox{1\linewidth}{!}{
\begin{tabular}{l|rr|rr|rr|rr} 
    \toprule
    \multicolumn{1}{c}{~} & \multicolumn{2}{c}{w/o Attack} & \multicolumn{2}{c}{Lexical-Attack} & \multicolumn{2}{c}{Dipper-Attack} & \multicolumn{2}{c}{Translation-Attack} \\
    \midrule
    \multicolumn{1}{c}{~} & \multicolumn{1}{c}{value$\uparrow$} & \multicolumn{1}{c}{$\Delta$} & \multicolumn{1}{c}{value$\uparrow$} & \multicolumn{1}{c}{$\Delta$} & \multicolumn{1}{c}{value$\uparrow$} & \multicolumn{1}{c}{$\Delta$} & \multicolumn{1}{c}{value$\uparrow$} & \multicolumn{1}{c}{$\Delta$} \\
    \midrule
    \multicolumn{1}{c}{~} & \multicolumn{8}{c}{Open Text Generation} \\
    \midrule

Gamma  &  99.58$_{\pm 0.01}$ & \multicolumn{1}{c|}{-} & 78.69$_{\pm 0.08}$ & \multicolumn{1}{c|}{-} & 55.18$_{\pm 0.11}$ & \multicolumn{1}{c|}{-} & 58.33$_{\pm 0.18}$ & \multicolumn{1}{c|}{-} \\
Delta  &  94.24$_{\pm 0.05}$ & \multicolumn{1}{c|}{-} & 62.48$_{\pm 0.29}$ & \multicolumn{1}{c|}{-} & 52.47$_{\pm 0.10}$ & \multicolumn{1}{c|}{-} & 54.52$_{\pm 0.16}$ & \multicolumn{1}{c|}{-} \\
\midrule
Unigram  &  99.59$_{\pm 0.06}$ & \multicolumn{1}{c|}{-} & 99.31$_{\pm 0.22}$ & \multicolumn{1}{c|}{-} & \underline{83.60$_{\pm 4.75}$} & \multicolumn{1}{c|}{-} & \underline{90.99$_{\pm 2.47}$} & \multicolumn{1}{c|}{-} \\
KGW  &  \underline{99.87$_{\pm 0.01}$} & \multicolumn{1}{c|}{-} & 99.24$_{\pm 0.02}$ & \multicolumn{1}{c|}{-} & 77.43$_{\pm 0.48}$ & \multicolumn{1}{c|}{-} & 85.88$_{\pm 0.06}$ & \multicolumn{1}{c|}{-} \\
KGW-Pool  &  \textbf{99.90$_{\pm 0.00}$} & 0.03$_{\pm 0.01}$ & \textbf{99.74$_{\pm 0.00}$} & 0.50$_{\pm 0.02}$ & \textbf{84.10$_{\pm 0.89}$} & 6.67$_{\pm 1.37}$ & \textbf{92.15$_{\pm 0.14}$} & 6.27$_{\pm 0.19}$ \\
\midrule
EXP  &  99.45$_{\pm 0.02}$ & \multicolumn{1}{c|}{-} & 98.97$_{\pm 0.02}$ & \multicolumn{1}{c|}{-} & 73.60$_{\pm 0.49}$ & \multicolumn{1}{c|}{-} & 80.45$_{\pm 0.05}$ & \multicolumn{1}{c|}{-} \\
EXP-Pool  &  99.76$_{\pm 0.01}$ & 0.31$_{\pm 0.02}$ & \underline{99.57$_{\pm 0.00}$} & 0.60$_{\pm 0.02}$ & 80.42$_{\pm 0.94}$ & 6.83$_{\pm 1.18}$ & 90.65$_{\pm 0.04}$ & 10.20$_{\pm 0.09}$ \\
\midrule
ITS  &  95.48$_{\pm 0.03}$ & \multicolumn{1}{c|}{-} & 83.19$_{\pm 0.11}$ & \multicolumn{1}{c|}{-} & 59.24$_{\pm 0.13}$ & \multicolumn{1}{c|}{-} & 56.88$_{\pm 0.03}$ & \multicolumn{1}{c|}{-} \\
ITS-Pool  &  99.18$_{\pm 0.01}$ & 3.70$_{\pm 0.03}$ & 96.70$_{\pm 0.00}$ & 13.52$_{\pm 0.11}$ & 61.68$_{\pm 0.03}$ & 2.44$_{\pm 0.15}$ & 71.33$_{\pm 0.03}$ & 14.45$_{\pm 0.06}$ \\

    \midrule
    \multicolumn{1}{c}{~} & \multicolumn{8}{c}{Long-Form Question Answering} \\
    \midrule

Gamma  &  99.85$_{\pm 0.01}$ & \multicolumn{1}{c|}{-} & 80.93$_{\pm 0.17}$ & \multicolumn{1}{c|}{-} & 55.14$_{\pm 0.02}$ & \multicolumn{1}{c|}{-} & 61.69$_{\pm 0.09}$ & \multicolumn{1}{c|}{-} \\
Delta  &  97.99$_{\pm 0.06}$ & \multicolumn{1}{c|}{-} & 65.72$_{\pm 0.17}$ & \multicolumn{1}{c|}{-} & 52.75$_{\pm 0.12}$ & \multicolumn{1}{c|}{-} & 57.48$_{\pm 0.12}$ & \multicolumn{1}{c|}{-} \\
\midrule
Unigram  &  99.83$_{\pm 0.10}$ & \multicolumn{1}{c|}{-} & 99.63$_{\pm 0.19}$ & \multicolumn{1}{c|}{-} & \textbf{87.79$_{\pm 1.92}$} & \multicolumn{1}{c|}{-} & \underline{94.38$_{\pm 1.05}$} & \multicolumn{1}{c|}{-} \\
KGW  &  \textbf{99.97$_{\pm 0.00}$} & \multicolumn{1}{c|}{-} & 99.66$_{\pm 0.00}$ & \multicolumn{1}{c|}{-} & 81.05$_{\pm 0.23}$ & \multicolumn{1}{c|}{-} & 92.34$_{\pm 0.09}$ & \multicolumn{1}{c|}{-} \\
KGW-Pool  &  \underline{99.96$_{\pm 0.00}$} & -0.01$_{\pm 0.00}$ & \underline{99.75$_{\pm 0.02}$} & 0.09$_{\pm 0.02}$ & \underline{87.41$_{\pm 0.58}$} & 6.36$_{\pm 0.38}$ & 94.29$_{\pm 0.09}$ & 1.95$_{\pm 0.12}$ \\
\midrule
EXP  &  99.86$_{\pm 0.02}$ & \multicolumn{1}{c|}{-} & 99.70$_{\pm 0.04}$ & \multicolumn{1}{c|}{-} & 80.55$_{\pm 0.41}$ & \multicolumn{1}{c|}{-} & 91.17$_{\pm 0.06}$ & \multicolumn{1}{c|}{-} \\
EXP-Pool  &  99.92$_{\pm 0.01}$ & 0.06$_{\pm 0.02}$ & \textbf{99.83$_{\pm 0.01}$} & 0.13$_{\pm 0.04}$ & 85.87$_{\pm 0.32}$ & 5.32$_{\pm 0.69}$ & \textbf{96.05$_{\pm 0.05}$} & 4.88$_{\pm 0.11}$ \\
\midrule
ITS  &  97.96$_{\pm 0.06}$ & \multicolumn{1}{c|}{-} & 88.41$_{\pm 0.19}$ & \multicolumn{1}{c|}{-} & 63.12$_{\pm 0.10}$ & \multicolumn{1}{c|}{-} & 68.02$_{\pm 0.21}$ & \multicolumn{1}{c|}{-} \\
ITS-Pool  &  99.75$_{\pm 0.00}$ & 1.79$_{\pm 0.06}$ & 98.48$_{\pm 0.01}$ & 10.07$_{\pm 0.19}$ & 66.35$_{\pm 0.13}$ & 3.23$_{\pm 0.22}$ & 81.99$_{\pm 0.07}$ & 13.97$_{\pm 0.27}$ \\
    
\bottomrule
\end{tabular}
}
}
\end{table}
}

{
\begin{table}[!htb]
\caption{Imperceptibility of different watermarking methods on OPT-6.7B. $\Delta$ is the difference between watermarked texts and non-watermarked texts. The best and second-best results are highlighted in \textbf{bold} and \underline{underline}.}
\label{tab:imperceptibility-opt6.7b}
\renewcommand{\arraystretch}{1.25}
\setlength{\tabcolsep}{1pt} 
\scriptsize
\centering
\makebox[\textwidth][c]{
\resizebox{1.1\linewidth}{!}{
\begin{tabular}{l|rr|rr|rr|rr|rr} 
    \toprule
    \multicolumn{1}{c}{~} & \multicolumn{2}{c}{Glob-distinct2} & \multicolumn{2}{c}{Glob-distinct3} & \multicolumn{2}{c}{Group-distinct2} & \multicolumn{2}{c}{Group-distinct3} & \multicolumn{2}{c}{ppl} \\
    \midrule
    \multicolumn{1}{c}{~} & \multicolumn{1}{c}{value$\uparrow$} & \multicolumn{1}{c}{$\Delta$$\uparrow$} & \multicolumn{1}{c}{value$\uparrow$} & \multicolumn{1}{c}{$\Delta$$\uparrow$} & \multicolumn{1}{c}{value$\uparrow$} & \multicolumn{1}{c}{$\Delta$$\uparrow$} & \multicolumn{1}{c}{value$\uparrow$} & \multicolumn{1}{c}{$\Delta$$\uparrow$} & \multicolumn{1}{c}{value$\downarrow$} & \multicolumn{1}{c}{$\Delta$$\downarrow$}\\
    \midrule
    \multicolumn{1}{c}{~} & \multicolumn{10}{c}{Open Text Generation} \\
    \midrule

    Non-watermark  &  39.38$_{\pm 0.88}$ & 0.00$_{\pm 0.00}$ & 75.81$_{\pm 2.02}$ & 0.00$_{\pm 0.00}$ & 84.27$_{\pm 2.30}$ & 0.00$_{\pm 0.00}$ & 94.08$_{\pm 2.73}$ & 0.00$_{\pm 0.00}$ & \underline{6.84$_{\pm 0.01}$} & \underline{0.00$_{\pm 0.00}$} \\
    Gamma  &  39.91$_{\pm 0.01}$ & 0.54$_{\pm 0.86}$ & 76.97$_{\pm 0.02}$ & 1.16$_{\pm 2.00}$ & 85.55$_{\pm 0.02}$ & 1.28$_{\pm 2.32}$ & 95.62$_{\pm 0.02}$ & 1.54$_{\pm 2.74}$ & 6.84$_{\pm 0.01}$ & 0.00$_{\pm 0.02}$ \\
    Delta  &  39.89$_{\pm 0.03}$ & 0.51$_{\pm 0.84}$ & 76.98$_{\pm 0.04}$ & 1.17$_{\pm 1.98}$ & 85.58$_{\pm 0.03}$ & 1.30$_{\pm 2.28}$ & \underline{95.65$_{\pm 0.02}$} & \underline{1.57$_{\pm 2.70}$} & 6.85$_{\pm 0.01}$ & 0.01$_{\pm 0.01}$ \\
    Unigram  &  36.47$_{\pm 1.80}$ & -2.91$_{\pm 0.92}$ & 72.36$_{\pm 2.39}$ & -3.45$_{\pm 1.19}$ & 82.88$_{\pm 2.17}$ & -1.39$_{\pm 1.54}$ & 94.82$_{\pm 0.65}$ & 0.73$_{\pm 2.27}$ & 8.78$_{\pm 0.42}$ & 1.94$_{\pm 0.42}$ \\
    KGW  &  38.21$_{\pm 0.01}$ & -1.17$_{\pm 0.87}$ & 74.79$_{\pm 0.02}$ & -1.02$_{\pm 2.03}$ & 85.27$_{\pm 0.02}$ & 0.99$_{\pm 2.32}$ & 95.57$_{\pm 0.02}$ & 1.49$_{\pm 2.74}$ & 8.43$_{\pm 0.02}$ & 1.59$_{\pm 0.02}$ \\
    KGW-Pool  &  \textbf{41.87$_{\pm 0.05}$} & \textbf{2.49$_{\pm 0.84}$} & \textbf{79.65$_{\pm 0.00}$} & \textbf{3.84$_{\pm 2.02}$} & \textbf{87.41$_{\pm 0.00}$} & \textbf{3.14$_{\pm 2.30}$} & \textbf{96.73$_{\pm 0.02}$} & \textbf{2.64$_{\pm 2.74}$} & 8.83$_{\pm 0.03}$ & 2.00$_{\pm 0.03}$ \\
    EXP  &  32.03$_{\pm 0.06}$ & -7.35$_{\pm 0.91}$ & 61.91$_{\pm 0.39}$ & -13.90$_{\pm 2.26}$ & 72.96$_{\pm 0.26}$ & -11.31$_{\pm 2.18}$ & 81.92$_{\pm 0.16}$ & -12.16$_{\pm 2.65}$ & 6.87$_{\pm 0.02}$ & 0.03$_{\pm 0.01}$ \\
    EXP-Pool  &  \underline{39.94$_{\pm 0.00}$} & \underline{0.56$_{\pm 0.88}$} & \underline{77.00$_{\pm 0.04}$} & \underline{1.19$_{\pm 2.00}$} & 85.58$_{\pm 0.03}$ & 1.30$_{\pm 2.29}$ & 95.65$_{\pm 0.04}$ & 1.56$_{\pm 2.70}$ & 6.85$_{\pm 0.01}$ & 0.01$_{\pm 0.01}$ \\
    ITS  &  35.89$_{\pm 0.02}$ & -3.49$_{\pm 0.87}$ & 68.24$_{\pm 0.07}$ & -7.57$_{\pm 1.95}$ & 75.87$_{\pm 0.09}$ & -8.41$_{\pm 2.22}$ & 84.53$_{\pm 0.11}$ & -9.56$_{\pm 2.63}$ & \textbf{6.57$_{\pm 0.01}$} & \textbf{-0.27$_{\pm 0.01}$} \\
    ITS-Pool  &  39.90$_{\pm 0.00}$ & 0.53$_{\pm 0.88}$ & 76.99$_{\pm 0.00}$ & 1.18$_{\pm 2.02}$ & \underline{85.58$_{\pm 0.00}$} & \underline{1.31$_{\pm 2.30}$} & 95.64$_{\pm 0.00}$ & 1.56$_{\pm 2.73}$ & 6.85$_{\pm 0.00}$ & 0.01$_{\pm 0.01}$ \\

    \midrule
    \multicolumn{1}{c}{~} & \multicolumn{10}{c}{Long-Form Question Answering} \\
    \midrule

    Non-watermark  &  \underline{33.00$_{\pm 0.02}$} & \underline{0.00$_{\pm 0.00}$} & \underline{71.26$_{\pm 0.03}$} & \underline{0.00$_{\pm 0.00}$} & \textbf{87.14$_{\pm 0.05}$} & \textbf{0.00$_{\pm 0.00}$} & \textbf{97.09$_{\pm 0.02}$} & \textbf{0.00$_{\pm 0.00}$} & 8.83$_{\pm 0.02}$ & 0.00$_{\pm 0.00}$ \\
    Gamma  &  32.94$_{\pm 0.04}$ & -0.06$_{\pm 0.04}$ & 71.21$_{\pm 0.10}$ & -0.05$_{\pm 0.09}$ & 87.05$_{\pm 0.03}$ & -0.08$_{\pm 0.03}$ & 97.03$_{\pm 0.02}$ & -0.05$_{\pm 0.03}$ & 8.80$_{\pm 0.01}$ & -0.04$_{\pm 0.03}$ \\
    Delta  &  32.96$_{\pm 0.03}$ & -0.03$_{\pm 0.03}$ & 71.20$_{\pm 0.08}$ & -0.06$_{\pm 0.10}$ & 87.06$_{\pm 0.02}$ & -0.07$_{\pm 0.07}$ & 97.04$_{\pm 0.02}$ & -0.04$_{\pm 0.04}$ & 8.80$_{\pm 0.02}$ & -0.04$_{\pm 0.05}$ \\
    Unigram  &  29.19$_{\pm 2.01}$ & -3.80$_{\pm 2.02}$ & 65.08$_{\pm 2.82}$ & -6.18$_{\pm 2.84}$ & 82.46$_{\pm 2.24}$ & -4.67$_{\pm 2.23}$ & 95.06$_{\pm 0.63}$ & -2.02$_{\pm 0.62}$ & 10.44$_{\pm 0.79}$ & 1.61$_{\pm 0.79}$ \\
    KGW  &  31.29$_{\pm 0.10}$ & -1.70$_{\pm 0.09}$ & 68.05$_{\pm 0.14}$ & -3.21$_{\pm 0.11}$ & 86.00$_{\pm 0.05}$ & -1.14$_{\pm 0.01}$ & 96.52$_{\pm 0.03}$ & -0.57$_{\pm 0.03}$ & 10.84$_{\pm 0.01}$ & 2.00$_{\pm 0.03}$ \\
    KGW-Pool  &  \textbf{34.81$_{\pm 0.02}$} & \textbf{1.82$_{\pm 0.04}$} & \textbf{73.48$_{\pm 0.03}$} & \textbf{2.22$_{\pm 0.06}$} & 85.15$_{\pm 0.05}$ & -1.98$_{\pm 0.10}$ & 95.10$_{\pm 0.00}$ & -1.99$_{\pm 0.02}$ & 10.48$_{\pm 0.03}$ & 1.65$_{\pm 0.01}$ \\
    EXP  &  25.16$_{\pm 0.79}$ & -7.83$_{\pm 0.78}$ & 54.09$_{\pm 1.78}$ & -17.17$_{\pm 1.76}$ & 75.30$_{\pm 0.08}$ & -11.84$_{\pm 0.11}$ & 84.66$_{\pm 0.09}$ & -12.43$_{\pm 0.07}$ & \underline{8.69$_{\pm 0.20}$} & \underline{-0.14$_{\pm 0.22}$} \\
    EXP-Pool  &  32.93$_{\pm 0.06}$ & -0.07$_{\pm 0.05}$ & 71.12$_{\pm 0.25}$ & -0.14$_{\pm 0.22}$ & \underline{87.08$_{\pm 0.00}$} & \underline{-0.06$_{\pm 0.05}$} & \underline{97.05$_{\pm 0.01}$} & \underline{-0.04$_{\pm 0.03}$} & 8.80$_{\pm 0.01}$ & -0.04$_{\pm 0.04}$ \\
    ITS  &  29.35$_{\pm 0.04}$ & -3.64$_{\pm 0.05}$ & 62.67$_{\pm 0.10}$ & -8.60$_{\pm 0.10}$ & 77.37$_{\pm 0.10}$ & -9.76$_{\pm 0.06}$ & 85.92$_{\pm 0.13}$ & -11.17$_{\pm 0.12}$ & \textbf{8.38$_{\pm 0.00}$} & \textbf{-0.45$_{\pm 0.02}$} \\
    ITS-Pool  &  32.90$_{\pm 0.00}$ & -0.09$_{\pm 0.02}$ & 71.15$_{\pm 0.00}$ & -0.11$_{\pm 0.03}$ & 87.05$_{\pm 0.00}$ & -0.08$_{\pm 0.05}$ & 97.03$_{\pm 0.00}$ & -0.05$_{\pm 0.02}$ & 8.79$_{\pm 0.00}$ & -0.04$_{\pm 0.02}$ \\
\bottomrule
\end{tabular}
}
}
\end{table}
}

{
\begin{table}[!htb]
\caption{Efficacy and Robustness of different watermarking methods on OPT-6.7B evaluated with TPR@FPR=1\%. $\Delta$ is the performance boost brought by WaterPool. The best and second-best results are highlighted in \textbf{bold} and \underline{underline}.}
\label{tab:robustness-opt6.7b}
\renewcommand{\arraystretch}{1.1}
\setlength{\tabcolsep}{2pt} 
\scriptsize
\centering
\makebox[\textwidth][c]{
\resizebox{1\linewidth}{!}{
\begin{tabular}{l|rr|rr|rr|rr} 
    \toprule
    \multicolumn{1}{c}{~} & \multicolumn{2}{c}{w/o Attack} & \multicolumn{2}{c}{Lexical-Attack} & \multicolumn{2}{c}{Dipper-Attack} & \multicolumn{2}{c}{Translation-Attack} \\
    \midrule
    \multicolumn{1}{c}{~} & \multicolumn{1}{c}{value$\uparrow$} & \multicolumn{1}{c}{$\Delta$} & \multicolumn{1}{c}{value$\uparrow$} & \multicolumn{1}{c}{$\Delta$} & \multicolumn{1}{c}{value$\uparrow$} & \multicolumn{1}{c}{$\Delta$} & \multicolumn{1}{c}{value$\uparrow$} & \multicolumn{1}{c}{$\Delta$} \\
    \midrule
    \multicolumn{1}{c}{~} & \multicolumn{8}{c}{Open Text Generation} \\
    \midrule

    Gamma  &  95.46$_{\pm 0.07}$ & - & 16.08$_{\pm 0.44}$ & - & 2.40$_{\pm 0.12}$ & - & 2.93$_{\pm 0.09}$ & - \\
    Delta  &  70.85$_{\pm 0.11}$ & - & 7.34$_{\pm 0.26}$ & - & 2.10$_{\pm 0.04}$ & - & 2.60$_{\pm 0.05}$ & - \\
    \midrule
    Unigram  &  93.68$_{\pm 2.32}$ & - & 89.39$_{\pm 4.31}$ & - & \underline{23.62$_{\pm 13.20}$} & - & 36.52$_{\pm 14.15}$ & - \\
    KGW  &  \textbf{97.58$_{\pm 0.08}$} & - & 86.17$_{\pm 0.49}$ & - & 14.63$_{\pm 0.05}$ & - & 25.88$_{\pm 0.16}$ & - \\
    KGW-Pool  &  \underline{96.77$_{\pm 0.11}$} & -0.81$_{\pm 0.03}$ & \underline{93.30$_{\pm 0.08}$} & 7.13$_{\pm 0.42}$ & \textbf{23.66$_{\pm 1.03}$} & 9.03$_{\pm 1.06}$ & \underline{37.91$_{\pm 0.12}$} & 12.03$_{\pm 0.10}$ \\
    \midrule
    EXP  &  94.84$_{\pm 0.35}$ & - & 88.97$_{\pm 0.68}$ & - & 15.40$_{\pm 1.24}$ & - & 27.16$_{\pm 1.28}$ & - \\
    EXP-Pool  &  96.56$_{\pm 0.96}$ & 1.72$_{\pm 1.31}$ & \textbf{94.19$_{\pm 0.07}$} & 5.22$_{\pm 0.66}$ & 22.03$_{\pm 0.85}$ & 6.63$_{\pm 2.05}$ & \textbf{43.84$_{\pm 0.52}$} & 16.68$_{\pm 0.89}$ \\
    \midrule
    ITS  &  64.71$_{\pm 0.46}$ & - & 21.06$_{\pm 0.59}$ & - & 2.14$_{\pm 0.08}$ & - & 2.95$_{\pm 0.17}$ & - \\
    ITS-Pool  &  88.43$_{\pm 0.11}$ & 23.72$_{\pm 0.47}$ & 60.89$_{\pm 0.34}$ & 39.84$_{\pm 0.57}$ & 3.77$_{\pm 0.06}$ & 1.63$_{\pm 0.13}$ & 8.98$_{\pm 0.11}$ & 6.03$_{\pm 0.20}$ \\

    \midrule
    \multicolumn{1}{c}{~} & \multicolumn{8}{c}{Long-Form Question Answering} \\
    \midrule

    Gamma  &  98.33$_{\pm 0.04}$ & - & 20.20$_{\pm 0.39}$ & - & 2.22$_{\pm 0.05}$ & - & 4.34$_{\pm 0.12}$ & - \\
    Delta  &  89.08$_{\pm 0.16}$ & - & 12.08$_{\pm 0.10}$ & - & 2.11$_{\pm 0.13}$ & - & 4.22$_{\pm 0.15}$ & - \\
    \midrule
    Unigram  &  96.50$_{\pm 2.59}$ & - & 90.67$_{\pm 6.19}$ & - & \underline{29.91$_{\pm 13.55}$} & - & 42.00$_{\pm 19.57}$ & - \\
    KGW  &  \underline{99.29$_{\pm 0.01}$} & - & 92.69$_{\pm 0.30}$ & - & 17.59$_{\pm 0.26}$ & - & 41.33$_{\pm 0.24}$ & - \\
    KGW-Pool  &  \textbf{99.40$_{\pm 0.17}$} & 0.11$_{\pm 0.17}$ & \underline{97.25$_{\pm 0.17}$} & 4.56$_{\pm 0.23}$ & 28.17$_{\pm 0.84}$ & 10.57$_{\pm 0.83}$ & 45.67$_{\pm 2.06}$ & 4.33$_{\pm 1.88}$ \\
    \midrule
    EXP  &  98.66$_{\pm 0.08}$ & - & 96.11$_{\pm 0.21}$ & - & 22.41$_{\pm 1.82}$ & - & \underline{49.18$_{\pm 0.97}$} & - \\
    EXP-Pool  &  99.27$_{\pm 0.04}$ & 0.61$_{\pm 0.06}$ & \textbf{98.12$_{\pm 0.03}$} & 2.01$_{\pm 0.23}$ & \textbf{32.14$_{\pm 0.08}$} & 9.73$_{\pm 1.89}$ & \textbf{65.61$_{\pm 0.06}$} & 16.43$_{\pm 0.92}$ \\
    \midrule
    ITS  &  81.56$_{\pm 0.37}$ & - & 33.43$_{\pm 0.74}$ & - & 2.77$_{\pm 0.07}$ & - & 6.38$_{\pm 0.19}$ & - \\
    ITS-Pool  &  96.48$_{\pm 0.08}$ & 14.92$_{\pm 0.37}$ & 78.37$_{\pm 0.36}$ & 44.94$_{\pm 0.66}$ & 6.10$_{\pm 0.26}$ & 3.33$_{\pm 0.21}$ & 20.88$_{\pm 0.36}$ & 14.50$_{\pm 0.44}$ \\

\bottomrule
\end{tabular}
}
}
\end{table}
}

{
\begin{table}[!htb]
\caption{ROC-AUC of different watermarking methods on OPT-6.7B. $\Delta$ is the performance boost brought by WaterPool. The best and second-best results are highlighted in \textbf{bold} and \underline{underline}.}
\label{tab:robustness-opt6.7b-roc}
\renewcommand{\arraystretch}{1.1}
\setlength{\tabcolsep}{2pt} 
\scriptsize
\centering
\makebox[\textwidth][c]{
\resizebox{1\linewidth}{!}{
\begin{tabular}{l|rr|rr|rr|rr} 
    \toprule
    \multicolumn{1}{c}{~} & \multicolumn{2}{c}{w/o Attack} & \multicolumn{2}{c}{Lexical-Attack} & \multicolumn{2}{c}{Dipper-Attack} & \multicolumn{2}{c}{Translation-Attack} \\
    \midrule
    \multicolumn{1}{c}{~} & \multicolumn{1}{c}{value$\uparrow$} & \multicolumn{1}{c}{$\Delta$} & \multicolumn{1}{c}{value$\uparrow$} & \multicolumn{1}{c}{$\Delta$} & \multicolumn{1}{c}{value$\uparrow$} & \multicolumn{1}{c}{$\Delta$} & \multicolumn{1}{c}{value$\uparrow$} & \multicolumn{1}{c}{$\Delta$} \\
    \midrule
    \multicolumn{1}{c}{~} & \multicolumn{8}{c}{Open Text Generation} \\
    \midrule

    Gamma  &  99.33$_{\pm 0.01}$ & \multicolumn{1}{c|}{-} & 77.01$_{\pm 0.14}$ & \multicolumn{1}{c|}{-} & 55.20$_{\pm 0.17}$ & \multicolumn{1}{c|}{-} & 57.47$_{\pm 0.20}$ & \multicolumn{1}{c|}{-} \\
    Delta  &  92.84$_{\pm 0.05}$ & \multicolumn{1}{c|}{-} & 61.06$_{\pm 0.26}$ & \multicolumn{1}{c|}{-} & 52.39$_{\pm 0.07}$ & \multicolumn{1}{c|}{-} & 53.83$_{\pm 0.04}$ & \multicolumn{1}{c|}{-} \\
    \midrule
    Unigram  &  99.48$_{\pm 0.19}$ & \multicolumn{1}{c|}{-} & 99.12$_{\pm 0.38}$ & \multicolumn{1}{c|}{-} & \textbf{83.94$_{\pm 5.55}$} & \multicolumn{1}{c|}{-} & \underline{89.73$_{\pm 3.35}$} & \multicolumn{1}{c|}{-} \\
    KGW  &  \underline{99.77$_{\pm 0.01}$} & \multicolumn{1}{c|}{-} & 98.96$_{\pm 0.03}$ & \multicolumn{1}{c|}{-} & 76.73$_{\pm 0.17}$ & \multicolumn{1}{c|}{-} & 83.73$_{\pm 0.16}$ & \multicolumn{1}{c|}{-} \\
    KGW-Pool  &  \textbf{99.84$_{\pm 0.00}$} & 0.06$_{\pm 0.01}$ & \textbf{99.62$_{\pm 0.00}$} & 0.66$_{\pm 0.03}$ & \underline{83.88$_{\pm 0.71}$} & 7.14$_{\pm 0.76}$ & \textbf{90.85$_{\pm 0.09}$} & 7.13$_{\pm 0.23}$ \\
    \midrule
    EXP  &  99.12$_{\pm 0.03}$ & \multicolumn{1}{c|}{-} & 98.32$_{\pm 0.04}$ & \multicolumn{1}{c|}{-} & 71.72$_{\pm 0.49}$ & \multicolumn{1}{c|}{-} & 78.93$_{\pm 1.57}$ & \multicolumn{1}{c|}{-} \\
    EXP-Pool  &  99.46$_{\pm 0.21}$ & 0.34$_{\pm 0.18}$ & \underline{99.29$_{\pm 0.03}$} & 0.96$_{\pm 0.07}$ & 77.78$_{\pm 1.12}$ & 6.06$_{\pm 1.61}$ & 88.22$_{\pm 0.13}$ & 9.29$_{\pm 1.45}$ \\
    \midrule
    ITS  &  93.25$_{\pm 0.06}$ & \multicolumn{1}{c|}{-} & 79.89$_{\pm 0.10}$ & \multicolumn{1}{c|}{-} & 59.15$_{\pm 0.18}$ & \multicolumn{1}{c|}{-} & 55.34$_{\pm 0.19}$ & \multicolumn{1}{c|}{-} \\
    ITS-Pool  &  98.63$_{\pm 0.01}$ & 5.38$_{\pm 0.07}$ & 95.40$_{\pm 0.06}$ & 15.51$_{\pm 0.07}$ & 61.17$_{\pm 0.12}$ & 2.02$_{\pm 0.23}$ & 69.49$_{\pm 0.10}$ & 14.15$_{\pm 0.16}$ \\

    \midrule
    \multicolumn{1}{c}{~} & \multicolumn{8}{c}{Long-Form Question Answering} \\
    \midrule

    Gamma  &  99.81$_{\pm 0.00}$ & \multicolumn{1}{c|}{-} & 80.10$_{\pm 0.11}$ & \multicolumn{1}{c|}{-} & 54.91$_{\pm 0.06}$ & \multicolumn{1}{c|}{-} & 60.16$_{\pm 0.05}$ & \multicolumn{1}{c|}{-} \\
    Delta  &  97.74$_{\pm 0.02}$ & \multicolumn{1}{c|}{-} & 65.49$_{\pm 0.07}$ & \multicolumn{1}{c|}{-} & 52.54$_{\pm 0.12}$ & \multicolumn{1}{c|}{-} & 56.29$_{\pm 0.07}$ & \multicolumn{1}{c|}{-} \\
    \midrule
    Unigram  &  99.80$_{\pm 0.14}$ & \multicolumn{1}{c|}{-} & 99.57$_{\pm 0.26}$ & \multicolumn{1}{c|}{-} & \textbf{88.43$_{\pm 4.17}$} & \multicolumn{1}{c|}{-} & \underline{93.74$_{\pm 2.41}$} & \multicolumn{1}{c|}{-} \\
    KGW  &  \textbf{99.96$_{\pm 0.00}$} & \multicolumn{1}{c|}{-} & 99.57$_{\pm 0.00}$ & \multicolumn{1}{c|}{-} & 80.23$_{\pm 0.16}$ & \multicolumn{1}{c|}{-} & 91.00$_{\pm 0.08}$ & \multicolumn{1}{c|}{-} \\
    KGW-Pool  &  \underline{99.95$_{\pm 0.01}$} & -0.01$_{\pm 0.01}$ & \underline{99.73$_{\pm 0.01}$} & 0.17$_{\pm 0.00}$ & \underline{86.84$_{\pm 0.62}$} & 6.61$_{\pm 0.49}$ & 93.32$_{\pm 0.04}$ & 2.31$_{\pm 0.10}$ \\
    \midrule
    EXP  &  99.81$_{\pm 0.01}$ & \multicolumn{1}{c|}{-} & 99.54$_{\pm 0.03}$ & \multicolumn{1}{c|}{-} & 77.31$_{\pm 1.28}$ & \multicolumn{1}{c|}{-} & 89.34$_{\pm 0.34}$ & \multicolumn{1}{c|}{-} \\
    EXP-Pool  &  99.90$_{\pm 0.00}$ & 0.09$_{\pm 0.01}$ & \textbf{99.77$_{\pm 0.01}$} & 0.23$_{\pm 0.02}$ & 84.38$_{\pm 0.15}$ & 7.07$_{\pm 1.16}$ & \textbf{94.71$_{\pm 0.02}$} & 5.36$_{\pm 0.34}$ \\
    \midrule
    ITS  &  97.04$_{\pm 0.02}$ & \multicolumn{1}{c|}{-} & 86.25$_{\pm 0.04}$ & \multicolumn{1}{c|}{-} & 62.49$_{\pm 0.14}$ & \multicolumn{1}{c|}{-} & 65.94$_{\pm 0.10}$ & \multicolumn{1}{c|}{-} \\
    ITS-Pool  &  99.67$_{\pm 0.01}$ & 2.63$_{\pm 0.01}$ & 98.08$_{\pm 0.01}$ & 11.83$_{\pm 0.04}$ & 65.85$_{\pm 0.13}$ & 3.36$_{\pm 0.28}$ & 80.03$_{\pm 0.10}$ & 14.10$_{\pm 0.18}$ \\

\bottomrule
\end{tabular}
}
}
\end{table}
}


\end{document}